\documentclass[letterpaper,twocolumn,10pt]{article}
\usepackage{usenix2019_v3}

\usepackage{algorithm}
\usepackage{algorithmic}

\usepackage{amsfonts}
\usepackage{nicefrac}
\usepackage{microtype}
\usepackage{amsmath}
\usepackage{amsthm}

\usepackage{titling}
\usepackage{authblk}

\usepackage{threeparttable}
\usepackage{array}
\newcolumntype{C}[1]{>{\centering\let\newline\\\arraybackslash\hspace{0pt}}m{#1}}

\usepackage{caption} 
\usepackage{pifont}

\usepackage[utf8]{inputenc}

\usepackage[english]{babel}
\usepackage{cite}
\usepackage{bbm}
\usepackage{url}
\usepackage{graphicx}

\usepackage{caption}
\usepackage[caption=false]{subfig}

\usepackage{etex}

\usepackage{color}
\usepackage{framed}
\usepackage{gensymb}
\usepackage{textcomp}
\usepackage{longtable}
\usepackage{array}
\usepackage{psfrag}
\usepackage{multirow}

\usepackage{adjustbox}

\usepackage{cleveref}

\usepackage{footnote}

\usepackage{rotating}

\usepackage{xcolor,colortbl}

\usepackage{booktabs}

\usepackage{tikz-qtree}
\usepackage{etex}

\usepackage{enumitem,amssymb}
\newlist{todolist}{itemize}{2}
\setlist[todolist]{label=$\square$}

\graphicspath{{figures/}}

\newtheorem{theorem}{Theorem}

\newcommand{\draft}[1]{}

\pretitle{\begin{center} \LARGE \bf}
\posttitle{\par\end{center}\vskip 0.5em}
\preauthor{
    \begin{center}
    \normalsize \lineskip 0.5em%
    \begin{tabular}[t]{c}
}
\postauthor{\end{tabular}\par\end{center}}

\predate{\begin{center}\large}
\postdate{\par\end{center}}

\renewcommand\abstractname{\large Abstract \vskip 0.1in}

\title{Identifying and characterizing ZMap scans:\\a cryptanalytic approach}

\newcommand{\cmark}{\ding{51}}
\newcommand{\xmark}{\ding{55}}
\newcommand{\offset}{\mathit{offset}}

\author{Johan Mazel}
\author{R\'{e}mi Strullu}

\affil{
    \texttt{\normalsize johan.mazel@ssi.gouv.fr, remi.strullu@ssi.gouv.fr} 
}
\affil{ANSSI}

\date{}

\begin{document}
\maketitle

\begin{abstract}
Network scanning tools play a major role in Internet security.
They are used by both network security researchers and malicious actors to 
identify vulnerable machines exposed on the Internet.
ZMap is one of the most common probing tools for high-speed Internet-wide 
scanning.
We present novel identification methods based on the IPv4 iteration process 
of ZMap.
These methods can be used to identify ZMap scans with a small number of 
addresses extracted from the scan.
We conduct an experimental evaluation of these detection methods on synthetic, 
network telescope, and backbone traffic.
We manage to identify 28.5\% of the ZMap scans in real-world traffic.
We then perform an in-depth characterization of these scans regarding, for
example, targeted prefix and probing speed.
\end{abstract}

\section{Introduction}

Internet wide scanning tools are commonly used by privacy and security 
researchers and malicious actors.
Use cases range from anti-censorship techniques 
\cite{Khattak2016You,Wustrow2014Tapdance} or computer security 
\cite{Gasser2016Scanning} research, to commercial services 
\cite{Durumeric2015Search,Shodan} and malicious mass-exploitation 
\cite{Falliere2011Distributed}.
Leonard et al. \cite{Leonard2013Demystifying} proposed the seminal work 
regarding high-speed uniformly spread Internet-wide scanning. 
Durumeric et al. and Robert Graham then published two tools, ZMap 
\cite{Durumeric2013ZMap,Adrian2014Zippier} and Masscan \cite{Masscan}, that 
considerably eased Internet-wide scanning.
This lead to a constant use increase \cite{Mazel2017Profiling}.
Security administrators monitor their networks to detect attacks at several 
stages such as reconnaissance, exploitation or command and control.
Incident response teams analyze network traffic logs to determine root 
causes of compromise.
Network scans are reconnaissance events, and thus interest security 
administrators (as occurring events) and incident response team (as security 
incident root cause).
These actors, however, often have trouble understanding probing 
scope and purpose.

Previous work \cite{Durumeric2014Security,Doerr2016Scan,Mazel2017Profiling} 
analyzed ZMap usage in the wild.
These works use heuristics \cite{Durumeric2014Security,Doerr2016Scan}, 
signatures \cite{Mazel2017Profiling}, and the fact that ZMap sets the IP ID 
field to 54321, to identify ZMap traffic.
They also rely on privileged points of view on traffic: a network telescope 
slightly smaller than a /9 prefix for \cite{Durumeric2014Security}, and 
backbone traffic containing prefixes adding up to a /13 prefix in 
\cite{Mazel2017Profiling}.
Security administrators usually operate smaller networks.
It is thus more difficult for them to detect incoming probing traffic 
\cite{Leonard2012Stochastic}, and to assess its nature.

Our goal is to identify ZMap scans and associated internal characteristics
from short sequences of observed packets.
We propose new methods to identify ZMap scans with or without the IP ID 
fingerprint.
These identification methods also recover targeted prefix which can be used by 
administrators and incident response teams to determine probing 
purpose.
For example, administrators can quickly discard indiscriminate Internet-wide 
probing, and focus on scans that specifically target their network.
By recovering the internal state of ZMap, we can also predict future probed 
addresses, and block related scanning traffic.
Similarly, incident response teams can use our identification methods to 
investigate scanning that targeted a compromised host.
Our methods are thus useful for both \emph{real-time} use cases (e.g. scan 
detection and probed IP address prediction and blocking), and 
\emph{a posteriori} ones (e.g. network traffic forensics).
Finally, we provide an in-depth characterization of ZMap 
scans that goes beyond existing work
\cite{Durumeric2014Security,Doerr2016Scan,Mazel2017Profiling}, for example 
regarding scan progress, packet rates and unnecessary probing.

Our contribution is threefold.
First, we propose two crypto-analysis methods to identify ZMap scanning.
Second, we evaluate efficiency and computing cost using synthetic data and 
real-world traffic captured from both network telescope and backbone traffic.
Third, we provide an in-depth analysis of ZMap usage in the wild.
We thus identify misuses, such as private IP address probing and packet rate 
above upstream capacity, that waste network resources.

Our paper is structured as follows.
\Cref{sec:related_work} present existing work on probing and ZMap scan 
analysis.
\Cref{sec:overview} details relevant aspects of ZMap design and implementation.
\Cref{sec:attacks} describes our identification methods.
\Cref{sec:results} presents our identification and characterization results.

\section{Related work}
\label{sec:related_work}

Several works analyze scans in real-world traffic 
\cite{Allman2007IMC,Glatz2012IMC,Brownlee2012PAM}.
ZMap usage has been documented in network telescope \cite{Durumeric2014Security} 
and backbone data \cite{Mazel2017Profiling}.
In \cite{Durumeric2014Security}, Durumeric et al. characterize ZMap usage, 
packet rate, source IP location, targeted ports and scan coverage.
In \cite{Mazel2017Profiling}, Mazel et al. describe the rise of ZMap usage from 
2013 to 2016.

Doerr et al. \cite{Doerr2016Scan} follow an approach similar to ours.
They extract the internal state of ZMap to predict future probed IP addresses 
and block forecasted traffic.
To this end, they brute force the sequence of observed IP addresses collected 
in a /16 network telescope using GPUs.
Their work, however, exhibit several limitations.
First, they do not describe the blacklist and offset mechanisms in ZMap.
Second, they do not take into account packet reordering which is an exteremely common 
phenomenom in networks 
\cite{Bennett1999Packet,Paxson1997End,Wang2004Study,Jaiswal2007Measurement,Murray2012State}.
Third, they do not provide any characteristics regarding observed scans.

We improve the work of Doerr et al. \cite{Doerr2016Scan} by taking into 
acount ZMap's blacklist and offset mechanisms (see \Cref{sec:overview}).
We also consider the impact of packet reordering on our methods and provide 
a sampling methodology that alleviates its impact (see \Cref{subsec:sampling}).
By leveraging ZMap's internal state, we detail ZMap scans characteristics 
that were not addressed in \cite{Durumeric2014Security,Mazel2017Profiling},
such as blacklist usage, targeted prefix, packet rate, evidence of cooperation 
among sources, scan progress and probe visibility (see \Cref{sec:prefix} to 
\Cref{sec:emitted_packet_rate}).

\section{ZMap overview}
\label{sec:overview}

We describe the process used by ZMap to iterate over the IPv4 address space.
This process relies on modular arithmetic in a finite field $\mathbb{F}_{p}$ 
defined by a prime number $p$.
Mathematical background can be found in \cite{cohen2013course}.

The structure of the iterator can be decomposed into two parts: the internal 
state $s$ that is updated over the integers of $\mathbb{F}_{p}$ and the 
function $f$ that maps the internal state into an IPv4 address.

\begin{figure}[t!]
  \centering
  \subfloat[Original sharding \cite{Adrian2014Zippier}]{
    \centering
    \includegraphics[width=0.5\columnwidth]{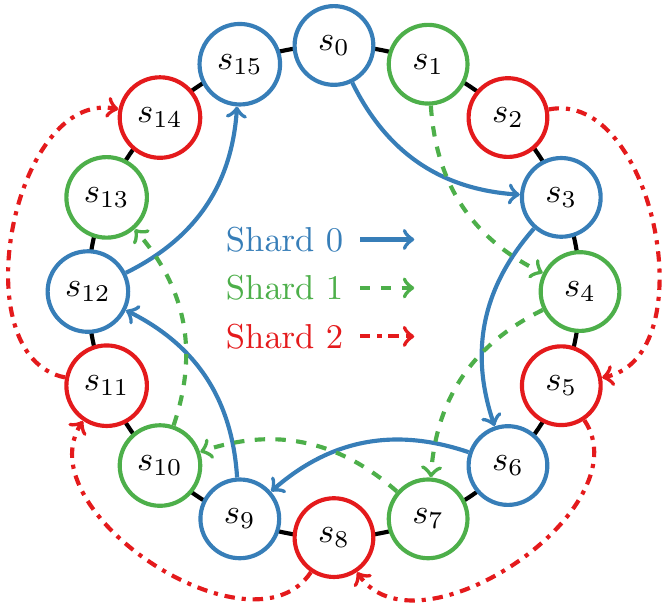}
    \label{fig:sharding_zippier}
  }
  \subfloat[New pizza-like sharding \cite{zmap_sharding_pizza}]{
    \centering
    \includegraphics[width=0.5\columnwidth]{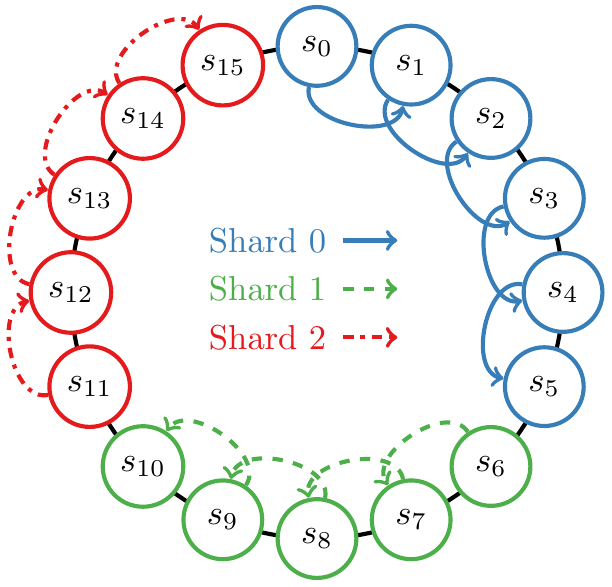}
    \label{fig:sharding_pizza}
  }
  \caption{Sharding mechanisms used in ZMap.
  }
  \label{fig:sharding}
\end{figure}

\subsection{Internal state}
\label{sec:internal_state}

At the initialization of the scan, a prime number $p$ is selected according to 
the number of IPv4 addresses to scan.
If we denote $n$ this number of addresses, then $p$ is the smallest prime number 
among the integers $\{p_0, ..., p_4\}$ 
 to be be greater or equal to $n$.
The prime numbers $p_0,..., p_4$ are predefined by the ZMap implementation: 
$p_0=2^{8}+1$, $p_1=2^{16}+1$, $p_2=2^{24}+43$, $p_3=2^{28}+3$ and 
$p_4=2^{32}+15$.

A primitive root $g$ of $\mathbb{F}_{p}$ and the initial internal state 
$s_0$ are also randomly generated at the beginning of each scan.
If the random generator has good randomness properties, we can suppose that 
these values are unique and characterize each scan.
Then, at each step $i$ in $\{1,..., p-2\}$, the state $s_i$ is updated 
according to the formula:
\begin{equation} \label{eq:update1}
s_{i} = g \cdot s_{i-1} \mod p
\end{equation}

ZMap supports distributed scans over several instances.
If this option is activated, the address generation process is 
parallelized between the instances by splitting the sequence of states into 
shards of equal length.
To perform this task, ZMap implements two types of sharding techniques that 
depend on the software version. 

Prior to a ZMap commit \cite{zmap_sharding_pizza} on September 15, 2017, the 
sharding technique was realized as follow : let $d$ be the number of instances 
that perform the distributed scan.
At the initialization step, a primitive root $g$ and a random number $s_{0}$ in 
$\mathbb{F}_{p}$ are shared between all the instances.
Then, each instance indexed by $k$ in $\{0,..., d-1\}$ generates its own 
initial internal state $s_{0, k} = s_{0} \cdot g^{k}$ and the internal state 
$s_{i, k}$ is updated according to the formula:
\begin{equation} \label{eq:update2}
s_{i, k} = g^d \cdot s_{i-1, k} \mod p
\end{equation}

If we denote $Shard_{k}$ the set of internal states generated by the instance 
$k$, then
$$Shard_{k} = \left\{s_{i,k} \mid i \in \left\{0,...,\Bigl \lfloor \dfrac{p-2}{d} \Bigr 
\rfloor \right\}\right\}$$ where $\Bigl \lfloor \dfrac{p-2}{d} \Bigr \rfloor$ is the 
integer division of $p-2$ by $d$.

It is easy to notice that 
$$\bigcup_{0 \leq k <d} Shard_{k} = \{s_i \mid i \in \{0,..., p-2 \}\} = \mathbb{F}_{p} 
\setminus \{0\}$$ and the sets $Shard_{k}$ are disjoint.
This allows to cover the whole address space to scan without repetition of an address 
into different shards.
This method is depicted in \Cref{fig:sharding_zippier}.

Since the ZMap version of September 15, 2017 \cite{zmap_sharding_pizza}, the 
sharding technique has been modified.
Each $Shard_{k}$ contains the consecutive states $s_i$ ($i$ in 
$\{k\cdot \lfloor p-2 / d \rfloor ,..., (k+1)\cdot \lfloor p-2 / d \rfloor - 1\}$) 
that would have been obtained with a non distributed scan
.
This is depicted in \Cref{fig:sharding_pizza}.
With our previous notations, the internal state of the instance $k$ is 
initialized with $s_{0,k}=s_0 \cdot g^{k\cdot \lfloor p-2 / d \rfloor}$ and is 
updated according to:
\begin{equation} \label{eq:update3}
s_{i, k} = g \cdot s_{i-1, k} \mod p
\end{equation}

In all these three settings, the choice of the values ($g$, $s_{0}$) is 
important as these values can be used to %
characterize the scan.
Two scans with the same couple of values $(g, s_0)$ will generate the same IPv4 
sequence of addresses.
In the case of distributed scans, it can also be used to identify the machines 
that cooperate to the same scan. %
ZMap generates these values pseudo-randomly by applying the block cipher AES to 
a selected seed.
By default, ZMap obtains the seed randomly from the Unix /dev/random device.
The user can, however, specify a seed as an option (-e) or in the configuration 
file.
Hence, when the user repeatedly launches ZMap with a specific seed, the 
same values ($g$, $s_{0}$) are generated and the same sequence of addresses 
will be probed.

\subsection{Mapping function}
The sequence of scanned IPv4 addresses is obtained by applying a mapping 
function $f$ to the current internal state at each step of the ZMap generation 
process.

The function $f$ is built at the initialization of the scan and depends on two 
parameters: a whitelist of IPv4 addresses $\mathcal{W}$ and a blacklist 
$\mathcal{B}$.
The whitelist $\mathcal{W}$ contains the addresses that the user has specified 
to scan and the blacklist $\mathcal{B}$ contains the addresses to skip.
Hence, if $\mathcal{S}$ is the set of IPv4 addresses that will eventually 
be scanned, we have the relation $\mathcal{S} = \mathcal{W} \setminus 
\mathcal{B}$.

The function $f$ is a one-to-one mapping from the integers $\{1,...,n\}$ to the 
address set $\mathcal{S}$.
The address sequence $(ip_1,...,ip_n)$ is then obtained by applying 
$f$ to the internal states $s_i$ when $s_i \leq n$, as described in \Cref{algo:algo1}.

\begin{algorithm}
\caption{IPv4 sequence generation $(ip_j)_{j\in \{1,...,n\}}$}
\begin{algorithmic}
	\STATE $j \leftarrow 0$
	\FOR{$i\in\{1,...,p\}$}
  		\STATE{Update $s_{i}$ using Eq.~\ref{eq:update1} (or 
        Eq.~\ref{eq:update2},\ref{eq:update3} in case of sharding)
      }
  		\IF{$s_i \leq n$}
  			\STATE $ip_j \leftarrow f(s_i)$
  			\STATE $j \leftarrow j+1$
  		\ENDIF
	\ENDFOR
\end{algorithmic}
\label{algo:algo1}
\end{algorithm}

The construction of $f$ relies on a tree structure $T$ that speeds up the 
computation of the addresses.
$T$ is a binary tree that is built at the initialization of the scan, when the 
address lists $\mathcal{W}$ and $\mathcal{B}$ are provided.
The leaves of $T$ represent disjoint subnets of $\mathcal{S}$ and the union of 
the leaves covers $\mathcal{S}$.
The network representation $t_l$ of a leaf $l$ in $T$ is the following: for 
each node, we assign the label $0$ to the left part of the tree under the node and $1$ to the 
right part.
With this coding, a leaf $l$ of $T$ can be represented as a finite sequence 
$(x_1, ..., x_q)$ of node selections where $x_i \in \{0, 1\}$.
Using the CIDR notation, we associate to $l$, the subnet
$t_l \subseteq \mathcal{S}$ whose prefix address is the bits $(x_1, ..., x_q)$ 
and the prefix length is $q$.
Furthermore, by construction of $T$, we enforce that, for each leaf 
$l$, $t_l$ is maximal in $\mathcal{S}$ in the sense that there is no subnet $t' 
\subseteq \mathcal{S}$ with $t_l \subseteq t'$ and prefix length smaller 
than $q$.

At each level $i$, the nodes of $T$ are labeled with 
the number of addresses in 
$\mathcal{S}$ that fall in the subnet of prefix length $i$ defined by the node.
When the building process of $T$ is finished, the mapping of an index $x \leq 
n$ 
to an address is performed by efficiently computing 
$f_0(x) = ip$ where $ip$ is the $x$th address in the list $\mathcal{S}$ ordered 
by the integer comparison of the host byte representation of the IPv4 
addresses.
This computation is performed by applying recursively \Cref{algo:algo2} to 
the tree $T$.

\begin{algorithm}
\caption{Index to address function $f_0(T,x)$}
\begin{algorithmic}
  	\IF{$T$ is a leaf $l$}
		\STATE $ip \leftarrow x_1 \cdot 2^{31} + x_2 \cdot 2^{30} + ... + x_q + x$ 
    where $(x_1, ..., x_q)$ is the sequence that codes the leaf
  		\STATE \textbf{return} the IPv4 address corresponding to $ip$ 
  	\ELSE
  		\STATE Let $T_0$ the left subtree of $T$ 
  		\STATE Let $T_1$ the right subtree 
  		\STATE Let $n_0$ 
      the label of the left node of $T$ (number of addresses in the left subtree)
	  	\IF{$x \leq n_0$}
  			\STATE \textbf{return} $f_0(T_0,x)$
  		\ELSE
  			\STATE \textbf{return} $f_0(T_1,x-n_0)$  		
  		\ENDIF
	\ENDIF
\end{algorithmic}
\label{algo:algo2}
\end{algorithm}

\subsection{Radix table}
In the last versions of ZMap (starting from v1.1.0 which was released on 
November 18th 2013), 
a lookup table called the radix table is used for 
speeding up the computation of the addresses in addition to the tree $T$.
The radix table $R$ contains the CIDR addresses of all the subnets of prefix 
length 20 included in $\mathcal{S}$.
For these versions of ZMap, the mapping function $f$ can be decomposed into two 
subfunctions $f_0$ and $f_1$.
If the index $x \in \{1,...,n\}$
 is less than $n_R \cdot 2^{12}$, where $n_R$ is the 
length of the radix table, then the radix table $R$ is used for the computation 
and $f(x) = f_1(R, x)$.
On the other hand, if $x>n_R \cdot 2^{12}$, then the tree $T$ built 
on the remaining addresses of $\mathcal{S}$ is used and $f(x) = f_0(T,x - n_R 
\cdot 2^{12})$.

The radix table $R$ contains the CIDR addresses of all the subnets of prefix 
length 20 included in $\mathcal{S}$ and is ordered by the integer comparison of 
the prefix in their host byte representation.
If $i$ is an integer smaller than $n_R$, $R[i]$ is the $i$th subnet of prefix length 20 included in $\mathcal{S}$.
Then $f_1(R,x)$ is defined as the $j$th address of the subnet $R[i]$ with $i$ 
the integer division of $x$ by $2^{12}$ and $j = x \mod 2^{12}$ (see \Cref{algo:algo3}).

\begin{algorithm}
\caption{Index to address function $f_1(R,x)$}
\begin{algorithmic}
	\STATE $i \leftarrow \Bigl \lfloor \dfrac{x}{2^{12}} \Bigr \rfloor $
	\STATE $j \leftarrow x - i \cdot 2^{12}$
	\STATE $ip \leftarrow R[i] \cdot 2^{12} + j$
  	\STATE \textbf{return} the IPv4 address corresponding to $ip$  		
\end{algorithmic}
\label{algo:algo3}
\end{algorithm}

\section{ZMap detection methods}
\label{sec:attacks}
We present two detection methods of ZMap scans based on the cryptanalysis of 
the IPv4 address generation.
These methods find the generator of the scan $g$ given a few samples of IPv4 
addresses extracted from the scan.

\subsection{Observation network}
\label{subsec:observation_network}
We denote by $\mathcal{O}$ the set of IPv4 addresses that are monitored by the 
administrator of a network.
In the case of an Internet-wide scan, some of the scanned addresses of 
$\mathcal{S}$ will reach the observed network $\mathcal{O}$.
We then make several hypotheses.

First, the administrator can retrieve a significant part of these addresses, 
and determine their arrival order.
Notice that we do not assume that all the scanned addresses that lay in the 
set $\mathcal{O}$ are observed by the administrator: some loss can occur 
during the capture of the packets.
However, the capture quality (packet loss rate and packet reordering) 
and the monitored network size both affect the performance 
(speed and success) of our detection methods.

Second, we also suppose that the observation network is completely included in 
a subnet of the scanned addresses.
Formally, if we order the monitored addresses by their host byte 
representation, we suppose that there is no IPv4 address 
$ip \not\in \mathcal{S}$ such that $o_1 < ip < o_2$ with $o_1, o_2$ in 
$\mathcal{O}$.
This occurs when $\mathcal{O}$ is a single subnet and there is no blacklisted 
address in $\mathcal{O}$.
Durumeric et al. \cite{Durumeric2014Security} reported that they received 
exclusion requests for 5.4 millions IP addresses.
They thus added 21,094 subnets of prefix length 24 to the default blacklist 
\cite{zmap_blacklisting}. 
This represents 0.1\% of all subnets of prefix length 24.
It means that if $\mathcal{O}$ is a subnet of prefix length 16 (resp. 20), and 
we consider the worst case where the blacklisted subnets are uniformly 
distributed among the remaining subnetworks of prefix length 24, then, the 
probability that $\mathcal{O}$ fulfill our hypothesis is 72\% (resp. 98\%).
This hypothesis on the observation network can greatly simplify the 
computation of the mapping function described in \Cref{sec:overview} 
as we show in the following Theorem.

\begin{theorem}\label{th:offset}
Suppose that there is a subnet $\mathcal{I}$ such that
$\mathcal{O}\subseteq \mathcal{I} \subseteq \mathcal{S}$.
Then
there is an integer that we denote $\emph{offset}$ such that for all 
$x \in \{1,...,n\}$ with $f(x) \in \mathcal{O}$

$$f(x) = hton(x + \mathit{offset})$$ where $hton$ 
is the function that converts host bytes to IPv4 addresses.
\end{theorem}

\begin{proof}
Let first assume that there exists a subnet $\mathcal{I'}$ of size $2^{12}$ 
with $\mathcal{O} \subseteq \mathcal{I'} \subseteq \mathcal{S}$.
If ZMap supports radix table then $f(x) = f_1(R, x)$.
Since $f(x)$ is in $\mathcal{O}$, then $f_1(R, x)$ is in $\mathcal{I'}$.
We deduce that there is a $i_0$ such that $R[i_0] = \mathcal{I'}$ and $\lfloor 
x /2^{12} \rfloor = i_0$.
From the definition of $f_1$, $f_1(R, x)$ is the $j$th address of the subnet 
$R[i_0]$. Hence we have $f(x) = hton(N + j)$ with $N$ the host byte 
representation of the subnet $R[i_0]=\mathcal{I'}$.
Then $f(x) = hton((N-i_0 \cdot 2^{12}) +( i_0 \cdot 2^{12} + j)) = hton(\offset 
+ x)$ with $\offset = N-i_0 \cdot 2^{12}$.

Now assume that ZMap still supports radix table, but the minimal size of the 
subnets $\mathcal{I'} \supseteq \mathcal{O}$ is greater than $2^{12}$.
Let $\mathcal{I''}$ such a minimal network.
Then
$\mathcal{I''} = \mathcal{I}_0 
\cup ... \cup \mathcal{I}_k$ where $\mathcal{I}_0,...,\mathcal{I}_k$ are the 
subnets of $\mathcal{I''}$ of size $2^{12}$ ordered by their prefix.
From the previous part of the proof, there exist $\offset_0$,...,$\offset_k$ 
such that $f(x) = hton(\offset_j + x)$ when $f(x)$ is in $\mathcal{I}_j$.
However, it is easy to notice that if we write $\offset_j = N_j - i_j \cdot 
2^{12}$, then $N_j = N_0 + j\cdot 2^{12}$ and $i_j = i_0+j$. Therefore 
$\offset_j = \offset_0$ for all $j \in \{1,...,k\}$.

Finally, assume that ZMap does not support radix computation or that the 
maximal size of the subnet $\mathcal{I'} \supseteq \mathcal{O}$ is less than 
$2^{12}$.
Then we have $f(x) = f_0(T, x)$.
Since $\mathcal{I}$ is a subnet of $\mathcal{S}$, $\mathcal{I}$ is included in 
a subnet $\mathcal{I'}$ represented by a leaf $l$ of $T$.
By definition $f_0(T, x)$ is the $j$th address of the subnet $\mathcal{I'}$, 
with $j = x - x_0$, and $x_0$ is the index of the first address of 
$\mathcal{I'}$.
Hence we have $f_0(T, x) = hton(N + j)$ with $N$ the host byte representation 
of the prefix $\mathcal{I'}$.
Therefore $f(x) = hton((N-x_0) + x) = hton(\offset + x)$ with $\offset = N-x_0$.
\end{proof}

At first glance, the second and the last hypotheses of the proof lead to the 
same result and it seems useless to fulfill the strongest hypothesis that 
$\mathcal{I}$ has size $\geq 2^{12}$.
There is, however, a difference that has an impact on the detection: when the 
monitored network is included in a subnet of size $\geq 2^{12}$ with no 
blacklisted address and ZMap supports radix table, then there are only $2^{20}$ 
possible values of $\emph{offset}$.
If the detection method requires to compute the list of indices from the observed 
addresses, then we may have to test $2^{20}$ offsets, instead of $2^{32}$ 
values in the last hypothesis of the proof.

\subsection{Detecting local scans}
\label{sec:detecting_local_scans}

The first detection method that we present can 
be applied to the scans for which $\mathcal{S} \subseteq \mathcal{O}$.
This case occurs when the observation network $\mathcal{O}$ is a large subnet 
of Internet (e.g. with prefix length 16), or if we want to detect small scans 
that are local to the observation network (e.g. scans of prefix length 24 in an 
observation network of prefix length 22).

We also assume that the size of the scan $n$ is close to the prime number $p$ 
used at the initialization of the scan.
This means that the subnet specified by the user (i.e. the 
whitelist $\mathcal{W}$) is a subnet of prefix length $8,16,24,28$ or $32$ and 
the blacklist $\mathcal{B}$ is disproportionately small compared to the size of 
the whitelist.

Let $o_1,o_2$, and $o_3$ be three consecutively observed IPv4 addresses.
As described in \Cref{sec:overview}, there is a sequence of internal 
states $s_{i_1}, s_{i_2}, s_{i_3}$ such that $f(s_{i_1})=o_1$, 
$f(s_{i_2})=o_2$, $f(s_{i_3})=o_3$.
From the hypothesis on the small difference between $n$ and $p$, there is a 
high probability that $s_{i_1+1}$ and $s_{i_1+2}$ are less than $n$.
Since $\mathcal{S} \subseteq \mathcal{O}$, we deduce that $f(s_{i_1+1})$, 
$f(s_{i_1+2})$ are in $\mathcal{O}$ and $i_2 = i_1 +1$, $i_3 = i_1 +2$, 
i.e. the internal states of consecutively observed addresses are also consecutive.
We can apply the following theorem to retrieve $g$.

\begin{theorem}
\label{th:tut}
Let $o_1, o_2, o_3$ three IPv4 addresses such that there exists an integer $i$ 
with $f(s_{i})=o_1$, $f(s_{i+1})=o_2$, $f(s_{i+2})=o_3$.
Then we have $$g = \frac{h_3 - h_1}{h_2 - h_1} - 1 \mod p$$ where $h_1$, $h_2$, 
$h_3$ are the host byte representations of $o_1$, $o_2$, $o_3$.
\end{theorem}

\begin{proof}
From \Cref{th:offset}, $h_{j} = s_{j} + \offset$ for $j$ in $\{1,2,3\}$.
By definition of the ZMap iteration process, we have $h_i \equiv s_1 \cdot 
g^{i-1} + \offset \pmod p$
where the symbol $\equiv$ denotes the congruence relation in the modular arithmetic.  
We deduce that $h_2 - h_1 \equiv s_1 \cdot (g-1) \pmod p$ and $h_3 - h_1 \equiv 
s_1 \cdot (g^2-1) \pmod p \equiv s_1 \cdot (g+1)\cdot(g-1) \pmod p$.
Since $g$ is a primitive root, $g \neq 1$ and we have the result.
\end{proof}

The description of the detection method is given in \Cref{algo:det1}.
Note that it does not require any hypothesis on the offset 
value of the mapping function $f$.
Once $g$ is known, it is possible to retrieve this value by computing $s_i = 
(h2-h_1)/(g-1) \mod p$ and $\mathit{offset} = h_1 - s_i$.
This gives additional information on the scanned addresses set $\mathcal{S}$ by 
using \Cref{th:offset}.

\begin{algorithm}
\caption{Detection method $Det1(p,(o_j)_{j\in \{1,...,m\}})$}
\begin{algorithmic}
  	\STATE Let $h_j$ the host byte representation of $o_j$ for $j\in\{1,...,m\}$
  	\STATE $g \leftarrow \frac{h_3 - h_1}{h_2 - h_1} - 1$
  	\STATE $s \leftarrow \frac{(h2-h_1)}{(g-1)} \mod p$
  	\STATE $\offset \leftarrow h_1 - s$  	
  	\FOR{$j \in \{4,..,m\}$}
  	  	\IF{$h_j \neq ((s \cdot g^j) \mod p + \offset)$}
  			\STATE \textbf{return} "not a ZMap scan"
  		\ENDIF
  	\ENDFOR
  	\STATE \textbf{return} $g, \offset$
 
\end{algorithmic}
\label{algo:det1}
\end{algorithm}

\subsection{Detecting Internet-wide scans}
\label{sec:detecting_internet_wide_scans}

We propose a second detection method that can be applied on scans whose scope 
is much wider than the observation network.
Even if the monitored network size has an impact on the computation time of the 
detection algorithm (see \Cref{subsec:offset_brute_forcing} and \Cref{sec:synthetic}), it is 
possible to detect internet-wide scans by using small monitored networks.
The method requires at least 20 packets for a high probability of success (see 
below).
Contrary to the previous method, it does not require consecutive packets.
In other words, packet loss does not impact the success of this method.

This detection method relies on the following Theorem:
\begin{theorem}\label{th:det2}
Let ($o_1,...,o_m$) a sequence of scanned IPv4 addresses ordered according to 
their packet arrival times,
and ($h_1,...,h_m$) their host byte representations.
For $b$ in $\mathbb{F}_{p}$, let $log_a(b)$ (or $log(b)$ for simplicity) the 
integer $x<p$ such that $b = a^x \mod p$ where $a$ denotes a predefined primitive 
root of $\mathbb{F}_{p}$ (i.e. $a=2$ or $a=3$ according to the value of 
$p \in \{p_0,...,p_4\}$).
If ($o_1,...,o_m$) is extracted from a ZMap scan, then there is an integer 
$k<p-1$ coprime with $p-1$ such that sequence ($e_2,...,e_m$) defined by
$$e_j = log\left(\frac{h_j-\offset}{h_1-\offset} \mod p\right) \cdot k^{-1} 
\mod p-1$$ is strictly increasing.

Furthermore, suppose that $k$ satisfies the property above,
and let $r = gcd(e_2,...,e_m)$ and $k' = k \cdot r$.
Then, with some additional assumptions on the discrete log distribution, we 
have
\begin{multline*}
\Pr[(o_1,...,o_m)  \; is \; a \; ZMap\; scan \; and \; k' = log(g)] \\ \geq    \left(1 - \frac{1}{(m-1)!}\right)^{\phi(p-1)} 
\end{multline*}
where $\Pr$ is the probability measure and $\phi$ is the Euler's totient function.
\end{theorem}
\begin{proof}
Let ($i_1,..., i_m$) the sequence of 
 state indexes such that $h_j = s_{i_j} + 
\offset = ((s_0 \cdot g^{i_j}) \mod p) + \offset$ for $j$ in $\{1,...,m\}$.
Let $\delta i_j = i_j - i_1$ for $j$ in $\{2,...,m\}$.
$\delta i_2,...,\delta i_m$ is an increasing sequence and if $k=log(g)$,
then $\delta i_j = e_j$.
This proves the existence of $k$.

Notice that if $\delta i_m \ll p$ and if $r$ is an integer with 
$r<\frac{p}{\delta i_m}$ and $r$ coprime with $p-1$, then the sequence $r \cdot 
\delta i_2 ,...,r \cdot\delta i_m$ is also increasing and $r \cdot \delta i_j = 
e_j$ for $k= log(g) \cdot r^{-1} \mod (p-1)$.

Let
\begin{multline*}
A = \{log(g) \cdot r^{-1} mod (p-1) \mid r<\frac{p}{\delta i_m}\; \\ and \; r \; coprime \; with \; p-1\}
.
\end{multline*} 
If $k\in A$, by assuming that $gcd(\delta i_2,...,\delta i_m)=1$ then we have $r = gcd(e_2,...,e_m)$ and therefore $log(g) = gcd(e_2,...,e_m) \cdot k$.

Now suppose that $k\not\in A$ and we want to estimate the probability that $k$ 
satisfies the property of the theorem ($A$ is the empty set if $o_1,...,o_m$ 
are not generated by a ZMap scan).
We assume that for any $k$, the sequence $(e_j)_{j\in \{2,...,m\}}$ defined by $k$ is equidistributed.
This assumption stems from the fact that discrete logarithm has an uniform 
distribution on arithmetic subsets of $\mathbb{F}_p$ when $p\rightarrow \infty$ 
as stated in \cite{gibson2012discrete}. Let $\Pr_k(e_2<...<e_m)$ the probality that the sequence ($e_2$,...,$e_m$) defined for a fixed $k$ is increasing.
Then for $k\not\in A, \Pr_k(e_2<...<e_m) = \frac{1}{(m-1)!}$ as there are 
$(m-1)!$ permutations of $\{2,...,m\}$ and each ordering of $(e_j)_{j\in 
\{2,...,m\}}$ can be defined by one of these permutations.
We also assume that the probabilities that $(e_j)_{j\in \{2,...,m\}}$ is not 
increasing are independent for (almost) all $k$.
Hence, the probability that for all $k\not\in A$, $(e_j)_{j\in \{2,...,m\}}$ is 
not increasing can be bounded from below by

\begin{multline*}
\prod_{k \; coprime \; with \; (p-1)} Pr_k[(e_j)_{j\in \{2,...,m\}}\; is \; not
\; \\ increasing] =  	 \left(1 - \frac{1}{(m-1)!}\right)^{\phi(p-1)}
.
\end{multline*}

Therefore, if $k$ defines an increasing sequence $(e_j)_{j\in \{2,...,m\}}$, then
$$Pr[k \in A] \geq  \left(1 - \frac{1}{(m-1)!}\right)^{\phi(p-1)}.$$

\end{proof}

The description of the ZMap detection algorithm is given in \Cref{algo:det2}.

\begin{algorithm}[H]
\caption{Detection method $Det2(p,\mathit{offset},(o_j)_{j\in \{1,...,m\}})$}
\begin{algorithmic}
  	\FOR{$j\in \{2,...,m\}$}
  		\STATE $f_j \leftarrow log_a\left(\frac{h_j-\offset}{h_1-\offset} \mod p\right)$ with $h_j$ the host byte representation of $o_j$ 
  	\ENDFOR
  	\FOR{$k \in \{1,...,p-1\}$ and $k$ coprime with $p-1$}
  		\FOR{$j \in \{2,..,m\}$}  		
  			\STATE $e_j \leftarrow f_j \cdot k$
  			\IF{$e_j \leq e_{j-1}$}
  				\STATE \textbf{break}
  			\ENDIF	
  			\IF{$j = m$}
  				\STATE $r \leftarrow gcd(e_2,...,e_m)$
  				\STATE \textbf{return} $g = a^{k \cdot r}$
  			\ENDIF
  		\ENDFOR
  	\ENDFOR
  	\STATE \textbf{return} "not a ZMap scan"
 
\end{algorithmic}
\label{algo:det2}
\end{algorithm}

The discrete logarithm $log_a$ can be efficiently computed by using baby-step 
giant-step algorithm \cite{shanks1971class}	or Pohlig-Hellman algorithm 
\cite{pohlig1978improved}.
The computation complexity of the discrete logarithm is $\mathcal{O}(\sqrt{p})$ 
for baby-step giant-step and $\mathcal O\left(\sum_i {y_i(\log n+\sqrt 
x_i)}\right)$ with $n=p-1=\prod_i x_i^{y_i}$ for Pohlig-Hellman.
Since the value of $p$ is at most $p_4=2^{32}+15$, we can expect to perform the 
computation respectively in $2^{16}$ or $2^{10}$ operations.
The iteration of $k$ amongst the numbers coprime with $p-1$ can also be 
efficiently performed by precomputing the list of such coprime numbers and 
storing them into a file.
During the main iteration loop of \Cref{algo:det2}, the file takes at most 
$4 \cdot \phi(p_4-1) =  4,5$ GB in memory.
Actually, the critical part of the algorithm is the multiplication of the 
$f_j$ by $k$ which is performed $j_m \cdot \phi(p-1)$ times in the worst case of 
a full iteration over the coprime integers ($j_m$ is the mean value of the least 
$j\leq m$ such that $e_j \leq e_{j-1}$, with the assumptions on the discrete 
log distribution of \Cref{th:det2} we can show that $j_m = e \approx 
2.72$).
This results in $j_m \cdot \phi(p-1) = 2.72 \cdot 1136578560 \approx 2^{32}$ 
operations with $p=p_4$.

The number of IPv4 addresses $m$ extracted from the scan has also an impact on 
the success of the detection method.
Let $\theta(m,p) = 1 - \left(1 - \frac{1}{(m-1)!}\right)^{\phi(p-1)}$ the upper 
bound of the failure probability of \Cref{th:det2} with $m$ addresses.
In the context of ZMap scan detection, this failure probability corresponds to 
the false positive rate.
For $p_4 = 2^{32}+15$, we have $\theta(14,p_4) = 0.16$, $\theta(15,p_4) = 
0.013$, $\theta(17,p_4) = 5.4 \cdot 10^{-5}$, and  $\theta(20,p_4) = 9.3 \cdot 
10^{-9}$.
As we test a high volume of scans in our experiments, we select $m=20$ for our 
implementation of the detection tool to avoid false positives.
Also note that, when the $\offset$ value is known, the false negative rate is 
equal to $0$.

\subsection{Offset computation}
\label{subsec:offset_computation}

The drawback of the detection method $Det2$ is that it requires to know the 
$\offset$ value of the generated IPv4 addresses (\Cref{subsec:observation_network}).
However, this value can be retrieved in two special cases.
The first case is when the scan does not use any blacklist (i.e. $\mathcal{B} = 
\emptyset$).
The second case is when the scan uses the default ZMap 
blacklist (i.e $\mathcal{B} = \mathcal{B}_{ZMap}$ \cite{zmap_blacklisting}).
Let $o_1,...,o_m$ be a sequence of IPv4 addresses to characterize.
For each whitelist $\mathcal{W}_k$ that corresponds to a subnet $t_k$ of prefix 
length $k$ containing $o_1,...,o_m$, we can compute the two corresponding 
$\offset$ values $\offset_{k}^{\mathcal{B}_{ZMap}}$ and $\offset_{k}^{\emptyset}$.
Then, for each $k \in \{0,...,24\}$, we apply $Det2$ with the inputs 
$o_1,...,o_m$, $\offset_{k}^{i}$ and $p_{r(k)}$ where $p_{r(k)}$ is the prime 
number associated to the size of the network $t_k$.
This overall detection method managed to identify and characterize a large 
amount of ZMap scans, as we will see later in our experiments (see 
\Cref{sec:results}).

\subsection{Offset brute forcing}
\label{subsec:offset_brute_forcing}

When ZMap uses a custom blacklist, the detection and the characterization of 
the scan can be performed by brute forcing the $\offset$ value.
If we assume that ZMap supports radix table for IPv4 computation as described 
in \Cref{sec:overview} and that the observation network $\mathcal{O}$ 
is included in a subnet $\mathcal{I}\subseteq \mathcal{S}$ of size $\geq 
2^{12}$ (which is the case when $\mathcal{O}$ is a subnet with no blacklisted 
IP and $\mathcal{S}$ is an Internet-wide scan), we have shown in 
\Cref{th:offset} that there are $2^{20}$ potential values for $\offset$.
Therefore, there are $2^{20}$ calls to $Det2$ for an exhaustive search of the 
offset value.
For each call of \Cref{algo:det2}, it is also possible to 
reduce the number of integers $k$ coprime with $p-1$ that are tested: 

\begin{theorem}\label{th:brute_force}
Assume that ($o_1,...,o_m$) is a sequence of ZMap probed IPv4 addresses in an 
observation network $\mathcal{O}$ which is included in the total scanned 
addresses $\mathcal{S}$.
Let $k$ a random number drawn uniformly from the integers in $\{1,...,p-1\}$ 
coprime with $p-1$ and $\Gamma$ the property on $k$ defined in \Cref{th:det2}.
Then $$Pr[k\text{ satisfies }\Gamma] \approx \frac{|\mathcal{O}|}{(m-1) \cdot 
|\mathcal{S}|}.$$ 
\end{theorem}
\begin{proof}
Let ($o_1,...,o_l$) a sequence of probed IPv4 addresses in the observation network 
and ($i_1,...,i_l$) the corresponding indices of the states, i.e. $o_j = 
hton(s_{i_j} + \offset)$ for all $j$ in $\{1, ..., l\}$.

If we denote $\gamma_{j} = i_{j} - i_{j-1}$ and 
$\gamma_{mean} = mean \{\gamma_{j} \mid 2 \leq j \leq l \}$, 
then we can notice that for $l$ 
sufficiently large, $\gamma_{mean} = \frac{|\mathcal{S}|}{|\mathcal{O}|}$.
Since the sequence of $\gamma_j$ can be viewed as a sequence of independent and 
identically distributed random variables of mean $\gamma_{mean}$ and $\delta i_m 
= i_m - i_1 = \gamma_2 + ... + \gamma_m$, by the law of large numbers we can 
approximate $\delta i_m$ by $\hat{\delta} i_m = (m-1) \cdot \gamma_{mean}$.
Let
$$A = \{log(g) \cdot r^{-1} \mid r<\frac{p}{\delta i_m}\; and \; r \; coprime \; with \; p-1\}.$$
As we have shown in \Cref{th:det2}, if $k$ is in $A$ then $k$ satisfies $\Gamma$,
and the probability that $k$ satisfies $\Gamma$ if $k$ is not in $A$ is very 
small.
Therefore
$$Pr[k\text{ satisfies }\Gamma] \approx Pr[k \in A].$$
From our hypothesis on the discrete logarithm distribution, we can assume that 
$A$ is equally distributed in the integers coprime with $p-1$.
Hence $$Pr[k \in A] = \frac{|A|}{\phi(p-1)} = \frac{1}{\delta i_m} \approx  
\frac{|\mathcal{O}|}{(m-1) \cdot |\mathcal{S}|}.$$
\end{proof}

\Cref{th:brute_force} underlines the fact that the size of the 
observation network $\mathcal{O}$ has an impact on the computation time of 
\Cref{algo:det2}.
If ($o_1,...,o_m$) is a sequence of ZMap probed IPv4 addresses then the mean 
number of modular multiplications $f_j \cdot k$ is equal to $\frac{e\cdot (m-1) 
\cdot |\mathcal{S}|}{|\mathcal{O}|}$.
For example, on an Internet-wide scan where 
$\mathcal{S} = \mathcal{S}_{Internet}=\{\text{"0.0.0.0"},...,"\text{255.255.255.255"}\}$, 
if the observation network $\mathcal{O}$ is a /24 
subnet and the number $m$ of ZMap probed addresses is equal to $20$, then the mean number of operations of  
\Cref{algo:det2} is approximately equal to $19 \cdot 2.72 \cdot 2^{24} = 2^{30}$, and if $\mathcal{O}$ is a /16 subnet then the mean number of operations is approximately equal to $2^{22}$.

\Cref{th:brute_force} also shows that we can improve the computation time 
of \Cref{algo:det2} by reducing the number of tested integers $k$ 
accordingly to a confidence threshold $\alpha$ which will correspond to the 
false negative rate. Let $T_{k_{max}}$ the first $k_{max}$ integers coprime with $p-1$.
By applying \Cref{th:brute_force}, if ($o_1,...,o_m$) 
is a ZMap scan and the integers $k$ tested in \Cref{algo:det2} are restricted to $T_{k_{max}}$, then
\begin{multline*}
1-\alpha = Pr[\exists k\in T_{k_{max}} \text{s.t. } 
k \text{ satisfies } \Gamma ] \\ \geq 1- \left(1 - \frac{|
\mathcal{O}|}{(m-1) \cdot |\mathcal{S}|} \right)^{k_{max}}
\end{multline*} where $\alpha$ is the false negative rate of the restricted algorithm.

By fixing the false negative rate $\alpha$, the minimum number of integers 
$k_{max}$ to test is equal to
$$k_{max} \approx \frac{log(\alpha)}{log(1-\frac{|\mathcal{O}|}{(m-1) \cdot |\mathcal{S}|})}.$$
If we apply the restricted algorithm $2^{20}$ times for an exhaustive search of the $\offset$ value, then the overall false negative rate is equal to $1-(1-\alpha)^{2^{20}}$ and the mean number of operations is equal to $k_{max} \cdot e \cdot 2^{20}$.

Let consider the case of the detection and the characterization of ZMap 
Internet-wide scans with a custom blacklist of IP addresses where $\mathcal{S} = 
\mathcal{S}_{Internet}$ and $\mathcal{O}$ is a /16 subnet.
If we want to achieve a false negative rate $\alpha=10^{-8}$ for the restricted 
version of \Cref{algo:det2}, then we can limit the number of integers coprime 
with $p-1$ to $k_{max} = 2^{25}$.
The overall false negative rate is equal to $0.01$ and the overall number of 
operations for each scan is equal to $2^{46}$.
This takes roughly $15$ days on a single core of a Intel Xeon E7-4820.

\subsection{Summary}

We summarize the detection methods' usages regarding scan types in \Cref{table:table1}.

\begin{table}[h!]
\centering
\captionof{table}{Scan types and detection methods.}
\label{table:table1}
\begin{tabular}{|p{4cm}|m{4cm}| }
  \hline
  \textbf{Scan types} & \textbf{Detection method} \\ \hline
  /24 scans or large observation network (>/16) & $Det1$  \\ \hline
  Scans with default or empty blacklists & $Det2$ with offset computation \\ \hline
  Scans with custom blacklists & $Det2$ with offset brute forcing  \\ \hline
\end{tabular}
\end{table}

\section{Results}
\label{sec:results}

This section presents our identification and characterization results using 
$Det2$ and pre-computed offsets on synthetic and real network traffic.
We do not run the $Det1$ method because the odds of observing three
consecutive packets inside our datasets (see \Cref{subsubsec:datasets}) are very low.
We also do not run the offset bruteforcing method because computing costs are 
too expensive (see \Cref{subsec:offset_brute_forcing}).
The method used hereafter in this section is $Det2$ with pre-computed offsets.

\subsection{Synthetic data}
\label{sec:synthetic}

In this section, we analyze the detection performance of the method $Det2$ 
using synthetic data.
We collect destination IP address sequences using the -d (or --dryrun) option of 
ZMap. 
For each prefix size between $0$ and $24$, we generate $100$ sequences of ZMap 
scans with different seeds.
Then, we restrict each of these sequences to the IPv4 addresses that lay into 
the observation network of the given prefix size.
We also ensure that we generate enough packets so that each scan contains $20$ 
packets in the observation network.

\subsubsection{Size of observed prefix $\mathcal{O}$ vs detection runtime}
\label{subsubsec:size_observed_prefix_vs_runtime}

In order to evaluate $Det2$ runtime, we measure the mean number of 
coprimes tested in the main loop of \Cref{algo:det2} for different 
observation network sizes.
We use a single thread to generate the packets.
The detection method $Det2$ identify all the ZMap scans and retrieve the 
correct generator $g$.

\begin{figure}[t!]
  \centering
  \includegraphics[width=1.0\columnwidth]{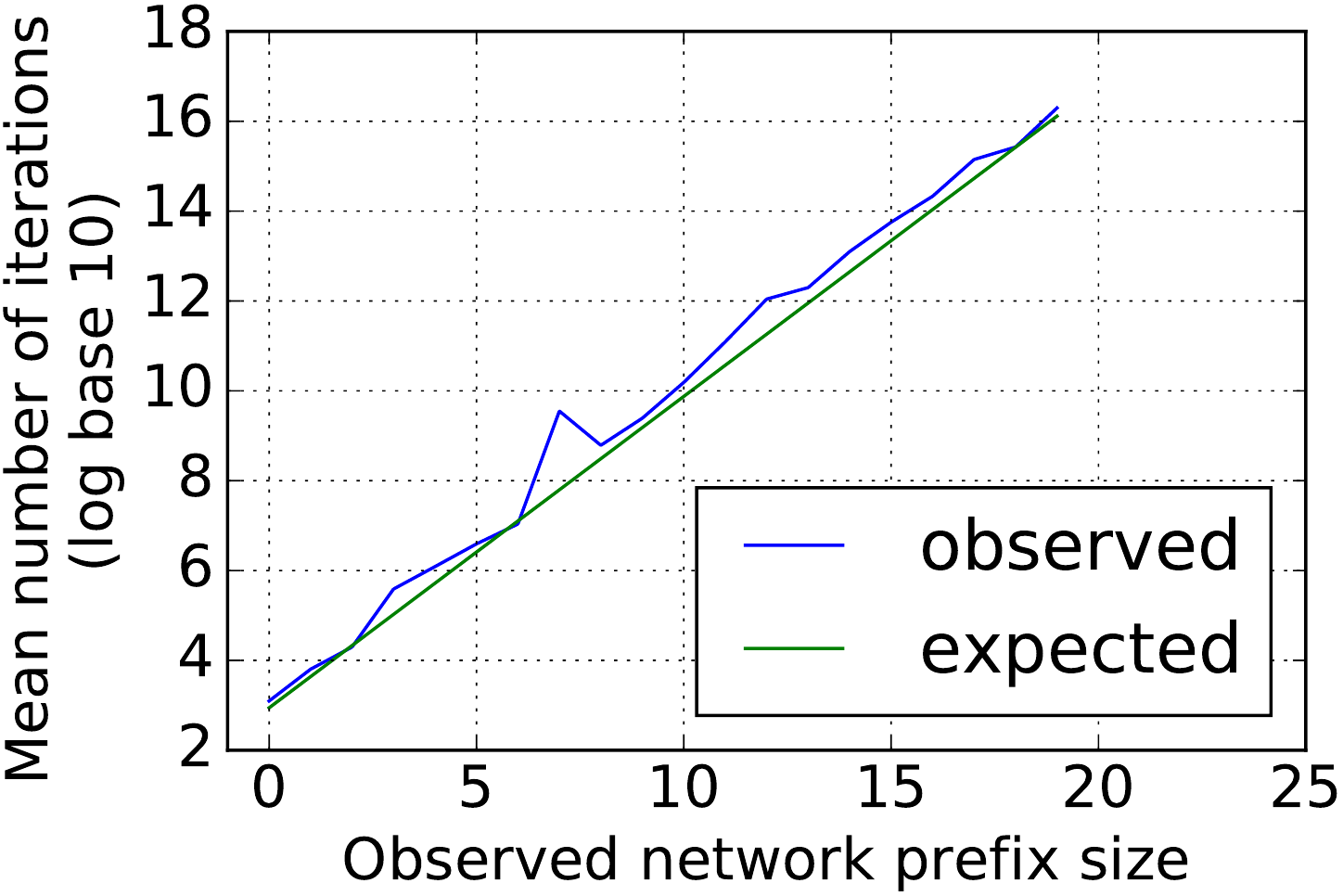}
  \caption{Mean number of coprimes tested for several sizes of observed prefix $\mathcal{O}$.}
  \label{fig:it_nb}
\end{figure}

\Cref{fig:it_nb} describes the mean number of iterations against the 
sizes of observed prefix $\mathcal{O}$.
We observe that the empirical results are similar to those obtained with the 
formula: $t(\mathcal{O}) = (m-1)\cdot \frac{|\mathcal{S}|}{|\mathcal{O}|}$.
This validates the approximation of the probability in \Cref{th:brute_force}.

\subsubsection{Number of threads vs detection success}
\label{subsec:threads}

\begin{figure}[t!]
  \centering
  \includegraphics[width=1.0\columnwidth]{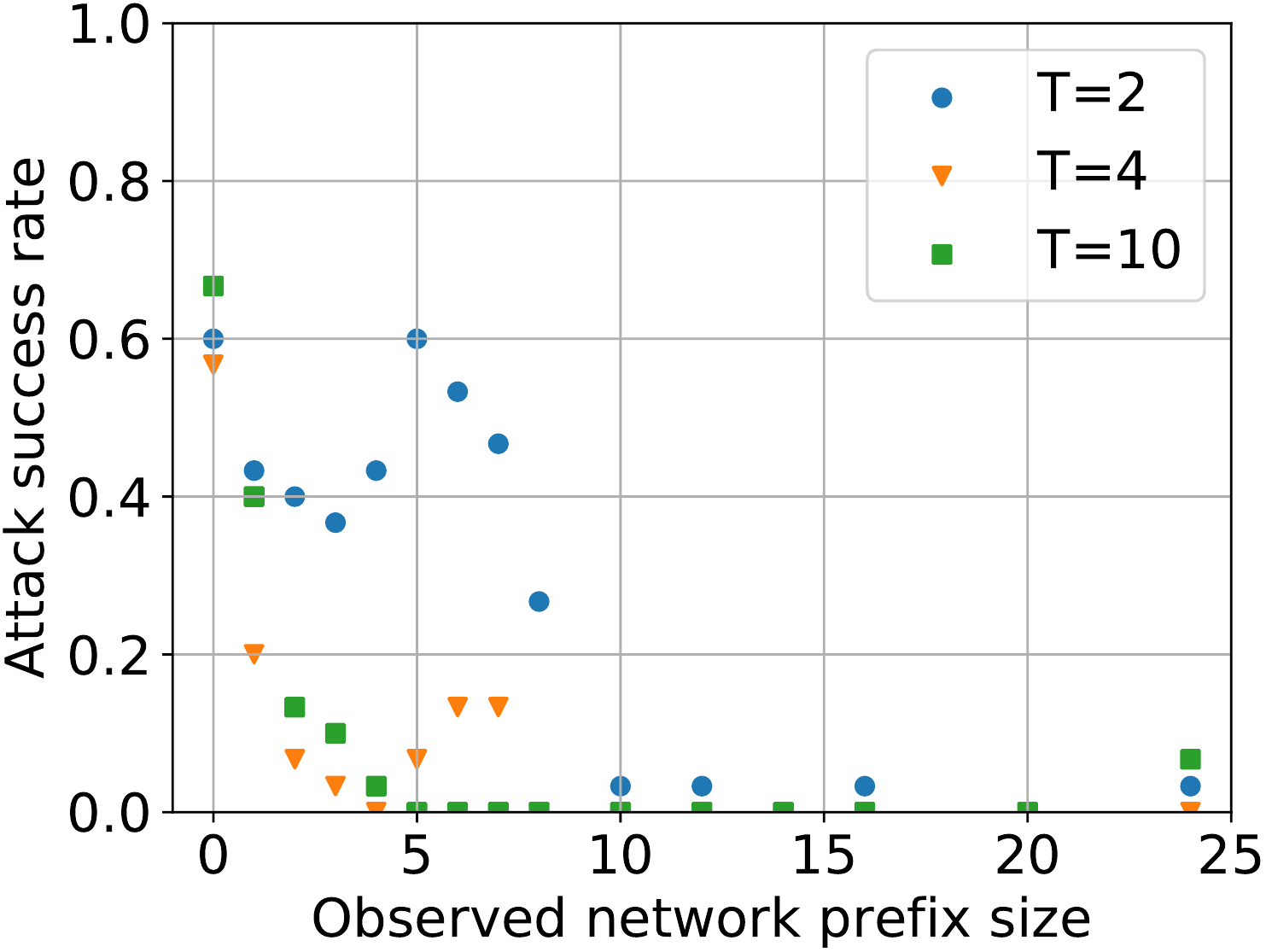}
  \caption{$Det2$ success rate for 2, 4 and 10 threads for several sizes of 
    observed prefix $\mathcal{O}$.}
  \label{fig:multithreading_success}
\end{figure}

ZMap uses an option (-T) to control the number threads that send probing packets.
Threads use sharding (see \Cref{sec:internal_state}) to share packet 
sending.
Using several threads may alter the packet sending order and thus impacts 
detection success.
Before December 6th, 2016 \cite{zmap_default_thread_number}, the default number 
of threads was the number of cores of the machine.
In the current version of ZMap (post \cite{zmap_default_thread_number}), the default 
number of thread is one.

\Cref{fig:multithreading_success} depicts the method success rate against the 
number of threads used for several sizes of observed prefix $\mathcal{O}$. 
We here use the first 20 packets observed as input to the $Det2$ method.
We observe that smaller numbers of threads yield higher success rates.
This is expected because when the thread number increases, the odds of 
observing a sequence from distinct threads that cannot be identified also 
increase.
Bigger prefix sizes yield a smaller success rate for all 
configurations.
The first 20 packets are quicker to obtain with small prefix sizes.
We thus hypothesize that, for small prefix sizes, observed packets are often sent by 
the first launched thread.
As the prefix size increases, the first 20 packets are increasingly sent by 
distinct threads which alter the original order.

\subsection{Real world traffic}
\label{subsec:real_world_traffic}

This section presents our identification and characterization results on real-world data.

\subsubsection{Datasets}
\label{subsubsec:datasets}

We use two real-world network traffic datasets.
The first dataset was collected using a network telescope from June 2014 to 
January 2015.
The observed prefix changes along time, from a single /24 to three
adjacent /24 prefixes inside a /22.
The second dataset is the MAWI repository \cite{mawi}.
This dataset is collected on a Japanese backbone link.
We use the multi-day long traces captured during the Day In The Life of 
Internet (DITL) event that took place on April 11th and 12th in 2017.
Unlike the publicly available MAWI traces where IP addresses are anonymized, 
we use original IP addresses.
We restrict our observation to a single /16 to fulfill the hypothesis on the 
observation network $\mathcal{O}$ from \Cref{subsec:observation_network}.

\subsubsection{Pre-processing}
\label{sec:pre_processing}

We identify scanning IP addresses using two methods.
First, we extract source IP address that have at least a packet with the ZMap 
fingerprint (IP ID field equals to 54321).
Second, we apply the Threshold Random Walk \cite{Jung2004SP} method using 
Zeek/Bro.
We then compute the union of source IP addresses extracted by the methods 
above, and keep only the IP addresses that send packets to at least 5 distinct 
destination IPs.
Individual scans performed by scanning IPs are then identified using a burst 
detection method \cite{Kleinberg2003Bursty}.
\Cref{table:data} presents detected scans.
We breakdown scans with more and less than 20 packets because the $Det2$ method 
needs at least 20 packets (see \Cref{sec:detecting_internet_wide_scans}).
Scan with at least 95\% of packets with the ZMap fingerprint (see above)
are labelled as ZMap in IP ID-related column hereafter.
Overall, we identify 1,314,804 scans, and apply the $Det2$ method on 
639,599 ones.

\begin{table}[t!]
  \setlength\tabcolsep{3pt}
  \centering
  
  \captionof{table}{Scan data breakdown.}
  
  \label{table:data}
  \begin{tabular}{llrrrrrrrrrr}
    \toprule
    \textbf{\# IP}
    & \textbf{IP ID}
    & \multicolumn{2}{c}{\textbf{Telescope}} 
    & \multicolumn{2}{c}{\textbf{Backbone}} \\
    
    \cmidrule(lr){3-4}
    \cmidrule(lr){5-6}
    
    &
    & \textbf{TCP} & \textbf{UDP}
    & \textbf{TCP} & \textbf{UDP} \\
        
    \midrule
    
    $<$ 20 IP  &        & 546,259 & 11,599 & 717,643 & 39,303 \\
    \midrule
    $>=$ 20 IP             
               & ZMap   &  12,200 &  6,944 &  17,737 &  2,869 \\
               & Other  & 197,421 &    612 & 391,901 &  9,915 \\
    
    \cmidrule(lr){2-6}
               & $\sum$ & 209,621 &  7,556 & 409,638 & 12,784 \\
    \bottomrule
  \end{tabular}
  
\end{table}

\subsubsection{Identification results}
\label{subsec:general}

\begin{table}[t!]
  \setlength\tabcolsep{1.5pt}
  \centering
  
  \captionof{table}{ZMap identification results. $\mathcal{B}_{ZMap}$ is the default list used by ZMap described at \cite{zmap_blacklisting}, $\emptyset$ means that the scan does not use any blacklist and $\mathcal{B}_{ZMap} \land \emptyset$ represents both default blacklist and lack thereof.
  }

  \label{table:results}
  \begin{tabular}{lccrrrrrrrrrr}
    \toprule
    \multirow{4}{*}{\rotatebox{90}{\textbf{IP ID field}}}
    & \multicolumn{2}{c}{\textbf{\textit{Det2}}}
    & \multicolumn{4}{c}{\textbf{Telescope}} 
    & \multicolumn{4}{c}{\textbf{Backbone}} \\
    
    \cmidrule(lr){2-3}
    \cmidrule(lr){4-7}
    \cmidrule(lr){8-11}
    
    & \multirow{2}{*}{\rotatebox{90}{\parbox{1.3cm}{\textbf{Result}}}}
    & \multirow{2}{*}{\rotatebox{90}{\parbox{1.3cm}{\textbf{Blacklist used}}}}
    & \multicolumn{2}{c}{\textbf{TCP}} & \multicolumn{2}{c}{\textbf{UDP}} 
    & \multicolumn{2}{c}{\textbf{TCP}} & \multicolumn{2}{c}{\textbf{UDP}} \\
    
    \cmidrule(lr){4-5}
    \cmidrule(lr){6-7}
    \cmidrule(lr){8-9}
    \cmidrule(lr){10-11}
    
    &
    &
    & \# & \% & \# & \% 
    & \# & \% & \# & \% \\
    \\[0.1cm]
        
    \midrule
    
    \multirow{7}{*}{\rotatebox{90}{ZMap}}
    
    & \xmark                           & &   7,139 & 58.5 & 6,285 & 90.5 &  12,353 & 69.6 & 2,671 & 93.1\\
    \cmidrule(lr){2-11}

    & \multirow{4}{*}{\cmark}
    & $\mathcal{B}_{ZMap}$               &   4,617 & 37.8 &   519 &  7.5 &   4,281 & 24.1 &   198 &  6.9 \\
    & & $\emptyset$                      &     422 &  3.5 &   140 &  2.0 &     972 &  5.5 &     - &    - \\
    & & \parbox{1.4cm}{$\mathcal{B}_{ZMap} \land \emptyset$}
                                         &      22 &  0.2 &     - &    - &     131 &  0.7 &     - &    - \\
    \cmidrule(lr){3-11}
    & & $\sum$                           &   5,061 & 41.5 &   659 &  9.5 &   5,384 & 30.4 &   198 &  6.9 \\
    
    \cmidrule(lr){2-11}
    & $\sum$ &                           &  12,200 &      & 6,944 &      &  17,737 &      & 2,869 &     \\
    
    \midrule
    \multirow{2}{*}{\rotatebox{90}{Other}}
    & \xmark &                           & 197,421 & 100 &   612 &   100 & 391,898 &  100 & 9,915 & 100 \\    
    \cmidrule(lr){2-11}
    & \cmark
    & $\mathcal{B}_{ZMap}$               &       - &   - &     - &     - &        3 &  0.0 &    - &   - \\    
    \bottomrule
  \end{tabular}
  
\end{table}

\Cref{table:results} presents the identification results.
For the telescope dataset, our method is successful on 37.8\% ZMap scans 
using TCP and 7.5\% using UDP.
For the backbone dataset, percentages of successful identification are lower: 
24.1\% (resp. 6.9\%) for TCP (resp. UDP) scans.
Overall, the majority of identified scans use the default ZMap blacklist \cite{zmap_blacklisting}.
Only 3.5\% (resp. 1.8\%) of ZMap TCP (resp. UDP) scans in the telescope data and 5.5\% 
of TCP scans in the backbone data do not use any blacklist.
This means that they actually send packets to IP addresses that are reserved 
such as private or multicast addresses.
They are thus needlessly increasing the workload of their infrastructure and 
their upstream provider.
We do not identify 58.5\% ZMap scans using TCP and 90.5\% 
using UDP in the telescope dataset, and 69.2\% ZMap scans using TCP and 93.1\% 
using UDP in the backbone dataset.
By analyzing reverse-DNS entries, we discover probing entities that confirm to 
us that they use customized blacklist.
In the telescope data, 60\% of UDP scans with ZMap's IPID are from Rapid7 
\cite{rapid7_project_sonar}, and that 11\% of TCP scans with ZMap's IPID are 
from University of Michigan \cite{umich}.
We thus hypothesize that, either unidentified scans use custom blacklists that 
we could not identify, or use several threads (see \Cref{subsec:threads}), or 
packet reordering occurred and our sampling strategy was not sufficient to cope 
with this problem (see \Cref{subsec:sampling}).
A small number of scans (22 in the telescope and 131 in the backbone) have 
been identified by our tool as using both the default ZMap blacklist, and no 
blacklist at the same time.
This is due to the fact that some of the precomputed offsets 
$\offset_{k}^{\mathcal{B}_{ZMap}}$ and $\offset_{k}^{\emptyset}$ (as described 
in \Cref{subsec:offset_computation}) are the same when the scanned network is small.
It is therefore impossible to determine if these scans use any blacklist.
The same generators are, however, returned in both cases, and it is still 
possible to infer additional scan characteristics.

We detect 3 scans \emph{without} the ZMap specific value for IP ID field in the backbone 
dataset.
These scans are performed by 3 IP in the same /24 prefix.
They have many common characteristics: they start and finish around the same 
time, they use the same generator $g$, exhibit 
low visibility (i.e. we do not see all their packets, see 
\Cref{subsec:visibility_progress}), target the whole Internet and the destination 
port is 443.
We thus hypothesize that a single entity uses sharding to decrease detection 
odds while removing ZMap IP ID field fingerprint to increase stealthiness.

\subsubsection{Run time}

\begin{table}[t!]
  \setlength\tabcolsep{3pt} %
  \centering
  \captionof{table}{ZMap identification run time (mean $\pm$ standard deviation in seconds). $\mu$ is the mean. $\sigma$ is the standard deviation. $\widetilde{x}$ is the median.}
  \label{table:execution_time}
  \begin{tabular}{llrrrrrrrrrrrr }
    \toprule

    \multirow{2}{*}{\rotatebox{90}{\textbf{Blacklist}}}
    & \multirow{2}{*}{\rotatebox{90}{\parbox{1.2cm}{\textbf{Result}}}}
    & \multicolumn{4}{c}{\textbf{Telescope}}
    & \multicolumn{4}{c}{\textbf{Backbone}}
    &  \\
    \cmidrule(lr){3-6}
    \cmidrule(lr){7-11}
    & 
    & $\mu$ & $\pm$ & $\sigma$ & $\widetilde{x}$
    & $\mu$ & $\pm$ & $\sigma$ & $\widetilde{x}$ \\
    \\
    \midrule

    \multirow{2}{*}{\rotatebox{90}{$\emptyset$}}
    & \xmark & 333.3 & $\pm$ & 100.9 & 290.6 & 178.5 & $\pm$ & 22.6 & 183.1 \\
    \cmidrule(lr){2-11}
    & \cmark & 104.6 & $\pm$ &  93.7 &  59.2 &  56.0 & $\pm$ & 48.6 &  38.4 \\
    
    \midrule

    \multirow{2}{*}{\rotatebox{90}{\parbox{0.95cm}{$\mathcal{B}_{ZMap}$}}}
    & \xmark & 351.3 & $\pm$ & 101.6 & 300.4 & 177.8 & $\pm$ & 22.6 & 182. \\
    \cmidrule(lr){2-11}
    & \cmark &  70.7 & $\pm$ &  38.1 &  58.5 &  57.1 & $\pm$ & 48.7 &  38.5  \\

    \bottomrule
  \end{tabular}

\end{table}

We instrument the execution time of our tool.
The network telescope dataset is analyzed using a machine with 4 Intel Xeon 
E7-4820 (octo core with Hyper-Threading).
The backbone dataset is analyzed using two machines with 2 Intel Xeon 
E5-2650 (octo core with Hyper-Threading).
In both cases, we restrict the number of paralellized executions to avoid 
Hyper-Threading.
Each scan is analyzed by our method using a single core.
\Cref{table:execution_time} details the measured execution time.
When the identification is successful, duration averages slightly more than a 
minute for the 
telescope dataset and slightly less than a minute for the backbone data.
When our method fails to identify a ZMap scan, the average run time is 
around 5 minutes for the telescope data and 3 minutes for backbone data.

\subsubsection{Impact of destination IP address sequence sampling}
\label{subsec:sampling}

Packet reordering occurs on devices \cite{Bennett1999Packet} or may be caused 
by 
routing events \cite{Mogul1992Observing}.
Network traffic analysis showed diverse percentage of packet reordering: 0.3 
to 2\% \cite{Paxson1997End}, 3.2\% \cite{Wang2004Study}, 1 to 1.5\% 
\cite{Jaiswal2007Measurement} and 0.074\% \cite{Murray2012State}.
These percentages are significant and potentially threaten the result of our 
identification method by inverting the order of consecutive observed packets $o_i$.
In order to alleviate this problem, we actually analyze both the raw sequence 
of the first 20 observed packet $o_i$ and a \emph{sampled} sequence.
This sequence contains 20 packets spread in the raw sequence: the first packet 
is $o_1$, the last one $o_c$ and the 18 other packets are evenly spread between 
$o_1$ and $o_c$.
\Cref{fig:sampling} provide an example of the sampling procedure for a scan 
with 41 packets.
We analyze both the raw sequence $\{o_1 \dots o_{20} \}$ and the sampled 
sequence which contain $o_i$ with $i={1 \dots 41}$ and $i$ an odd positive 
integer.

\begin{figure}[t!]
  \centering
  \includegraphics[width=1.0\columnwidth]{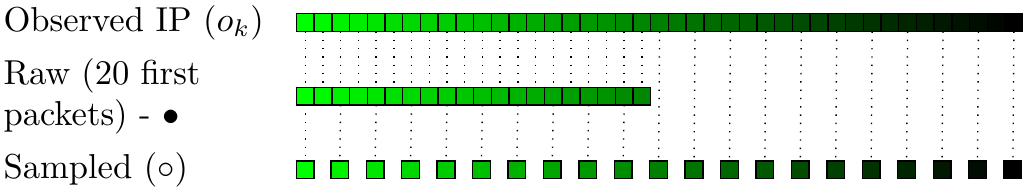}
  \caption{Example of raw ($\bullet$) and sampled ($\circ$) sequences for a scan containing 41 destination IP addresses.}
  \label{fig:sampling}
\end{figure}

\begin{table}[t!]
  \setlength\tabcolsep{4pt}
  \centering
  \captionof{table}{Identification results of scans with the ZMap fingerprint 
    regarding the raw ($\bullet$) and sampled ($\circ$) sequences.
  }
  \label{table:sampling}
  \begin{tabular}{lcrrrrrrrrrrrr }
    \toprule
    \multirow{3}{*}{\rotatebox{90}{\parbox{1.4cm}{\textbf{Blacklist}}}}
    & \multirow{3}{*}{\parbox{1.0cm}{\textbf{Data type}}}
    & \multicolumn{4}{c}{\textbf{Telescope}}
    & \multicolumn{4}{c}{\textbf{Backbone}} \\
    
    \cmidrule(lr){3-6}
    \cmidrule(lr){7-10}

    & 
    & \multicolumn{2}{c}{\textbf{TCP}} & \multicolumn{2}{c}{\textbf{UDP}}
    & \multicolumn{2}{c}{\textbf{TCP}} & \multicolumn{2}{c}{\textbf{UDP}} \\
    
    \cmidrule(lr){3-4}
    \cmidrule(lr){5-6}
    \cmidrule(lr){7-8}
    \cmidrule(lr){9-10}
    & 
    & \# & \% & \# & \% 
    & \# & \% & \# & \% \\ 
    \midrule
    
    \multirow{4}{*}{\rotatebox[origin=c]{90}{$\mathcal{B}_{ZMap}$}}
    & $\bullet \land \circ$
    & 4,513 & 97.7 & 444 & 85.5 & 3,975 & 92.9 &  16 &  8.1 \\
    & $\bullet \land \lnot \circ$ 
    &    64 &  1.4 &  60 & 11.6 &     8 &  0.2 &   1 &  0.5 \\
    & $\lnot \bullet \land \circ$ 
    &    40 &  0.9 &  15 &  2.9 &   298 &  7.0 & 181 & 91.4 \\
    \cmidrule(lr){2-10}
    & $\sum$              & 4,617 &      & 519 &      & 4,281 &      & 198 \\
    
    \midrule

    \multirow{4}{*}{\rotatebox{90}{$\emptyset$}}
    & $\bullet \land \circ$
    & 350 & 82.9 & 105 & 75.0 &   954 & 98.1 & - &   - \\
    & $\bullet \land \lnot \circ$ 
    &  36 &  8.5 &  20 & 14.3 &     2 &  0.2 & - &   - \\
    & $\lnot \bullet \land \circ$ 
    &  36 &  8.5 &  15 & 10.7 &    16 &  1.6 & - &   - \\
    \cmidrule(lr){2-10}
    & $\sum$              & 444 &      & 140 &      &   972 &      & - \\
    \midrule

    \multirow{4}{*}{\rotatebox{90}{\parbox{1.5cm}{$\mathcal{B}_{ZMap}$-$\emptyset$}}}
    
    & $\bullet \land \circ$
    &  1 &  4.5 & - &   - &   4 &  3.1 & - &   - \\
    & $\bullet \land \lnot \circ$
    &  4 & 18.2 & - &   - &  21 & 16.0 & - &   - \\
    & $\lnot \bullet \land \circ$
    & 17 & 77.3 & - &   - & 106 & 80.9 & - &   - \\
    \cmidrule(lr){2-10}
    & $\sum$              & 22 &      & - &     & 131 &      & - & \\
    
    \midrule
    
    \multirow{4}{*}{*}
    & $\bullet \land \circ$
    & 4,864 & 96.1 & 549 & 83.3 & 4,933 & 91.6 &  16 & 8.1 \\
    & $\bullet \land \lnot \circ$ 
    &   104 &  2.1 &  80 & 12.1 &    31 &  0.6 &   1 & 0.5 \\
    & $\lnot \bullet \land \circ$
    &    93 &  1.8 &  30 &  4.6 &   420 &  7.8 & 181 & 91.4 \\
    \cmidrule(lr){2-10}
    & $\sum$              & 5,061 &      & 659 &      & 5,384 &      & 198 &      \\
    
    \bottomrule
  \end{tabular}
  
\end{table}

\begin{table}[t!]
  \setlength\tabcolsep{3pt} %
  \centering
  \captionof{table}{Targeted prefixes of scans with the ZMap fingerprint with full percentages calculations.}
  \label{table:targeted_prefixes}
  \begin{tabular}{lrrrrrrrrrrr}
    \toprule
    \multirow{3}{*}{\rotatebox{90}{\parbox{1.4cm}{\textbf{Blacklist}}}}
    & \multirow{3}{*}{\parbox{1.0cm}{\textbf{Prefix size}}}
    & \multicolumn{4}{c}{\textbf{Telescope}}
    & \multicolumn{4}{c}{\textbf{Backbone}} \\
    \cmidrule(lr){3-6}
    \cmidrule(lr){7-10}
    
    & 
    & \multicolumn{2}{c}{\textbf{TCP}} & \multicolumn{2}{c}{\textbf{UDP}}
    & \multicolumn{2}{c}{\textbf{TCP}} & \multicolumn{2}{c}{\textbf{UDP}}
    \\
    
    \cmidrule(lr){3-4}
    \cmidrule(lr){5-6}
    \cmidrule(lr){7-8}
    \cmidrule(lr){9-10}
    & 
    & \# & \% & \# & \% 
    & \# & \% & \# & \% \\ 
    
    \midrule
    
    \multirow{3}{*}{\rotatebox{90}{\parbox{0.7cm}{$\mathcal{B}_{ZMap}$}}}
    
    & 0  & 4,615 & 91.2 & 518 & 78.6 & 4,284 & 79.6 & 198 & 100.0 \\
    & 1  &     2 &  0.0 &   1 &  0.2 &     - &    - &   - &     - \\
    
    \midrule
    
    \multirow{3}{*}{\rotatebox{90}{$\emptyset$}}
    
    & 0  &   414 & 8.2 & 140 & 21.2 &    972 & 18.1 &   - &    - \\
    & 1  &     5 & 0.1 &   - &    - &      - &    - &   - &    - \\
    & 3  &     3 & 0.1 &   - &    - &      - &    - &   - &    - \\
    
    \midrule
    
    \multirow{3}{*}{\rotatebox{90}{\parbox{1.cm}{$\mathcal{B}_{ZMap}$ \& $\emptyset$}}}
    
    &  5 &     1 & 0.0 &   - &    - &     - &    - &   - &    - \\
    &  8 &     4 & 0.1 &   - &    - &   131 &  2.4 &   - &    - \\
    & 16 &    17 & 0.3 &   - &    - &     - &    - &   - &    - \\
    
    \midrule
    
    \multirow{6}{*}{*}
    
    &  0 & 5,029 & 99.4 & 658 & 99.8 & 5,256 & 97.6 & 198 & 100.0\\
    &  1 &     7 &  0.1 &   1 &  0.2 &     - &    - &   - &     -\\
    &  3 &     3 &  0.1 &   - &    - &     - &    - &   - &     -\\
    &  5 &     1 &  0.0 &   - &    - &     - &    - &   - &     -\\
    &  8 &     4 &  0.1 &   - &    - &   131 &  2.4 &   - &     -\\
    & 16 &    17 &  0.3 &   - &    - &     - &    - &   - &     -\\
        
    \bottomrule
  \end{tabular}
\end{table}

\Cref{table:sampling} presents the breakdown of our results regarding sampling.
In the telescope data, scans that were only successfully identified with sampling 
represent a minority of all scans: 1.8\% for TCP and 4.6\% for UDP.
In backbone data, ``sampling only'' exhibits a bigger percentage of successfully 
identified scans: 7.8\% of TCP scans and 91.4\% of UDP ones.
We hypothesize that the bigger monitored prefix (65,535 IP addresses in the 
backbone data and 768 IP addresses in the telescope data) increases the 
odds of observing reordered packets.
Another possible explanation is that usual probing speed increased from 2015 
to 2017, and thus caused a packet reordering rise.
We also observe that some scan are identified with raw sequences but not with 
sampled sequences.
We hypothesize that this is due to time-splitting errors (see 
\Cref{sec:pre_processing}).
When this error happens, one or several scans using distinct generators are 
merged together into a single one.
This causes the sampled sequence to contain observed packets from one or more 
scans with distinct generators.
In that case, our method is not successful.

\subsubsection{Targeted prefix}
\label{sec:prefix}

\Cref{table:targeted_prefixes} details the prefix targeted by the identified 
ZMap scans.
More than 99\% of scans found in the telescope data, and more than 97\% of 
scans identified in the backbone data, target a /0 prefix.
This is expected since ZMap is one of the two main high speed scanning tools 
(with Masscan \cite{Masscan}) that aim to perform fast and large scale 
probing.
We however note that some scans target smaller prefixes (e.g. /16).

\begin{figure}[t!]
  \centering
  \includegraphics[width=1.0\columnwidth]{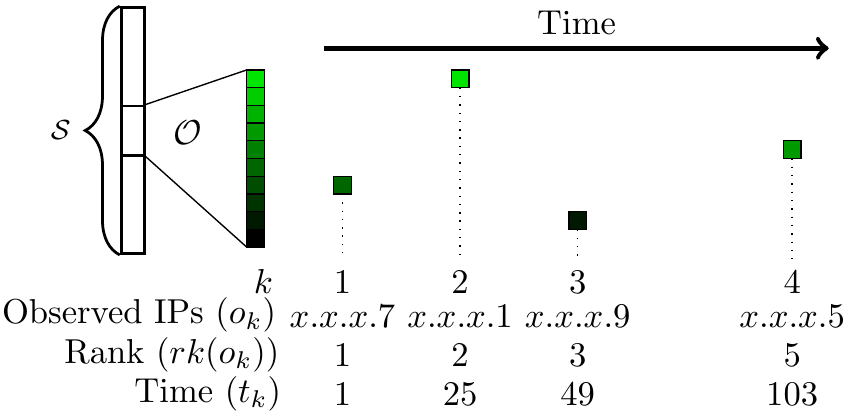}
  \caption{
    Example of scanned network $\mathcal{S}$ and observed network $\mathcal{O}$ 
    with missing packet for $rk = 4$ due to packet loss.
  }
  \label{fig:notations}
\end{figure}

\begin{figure*}[t!]
  \centering
  \subfloat[Telescope]{
    \centering
    \includegraphics[width=0.5\textwidth]{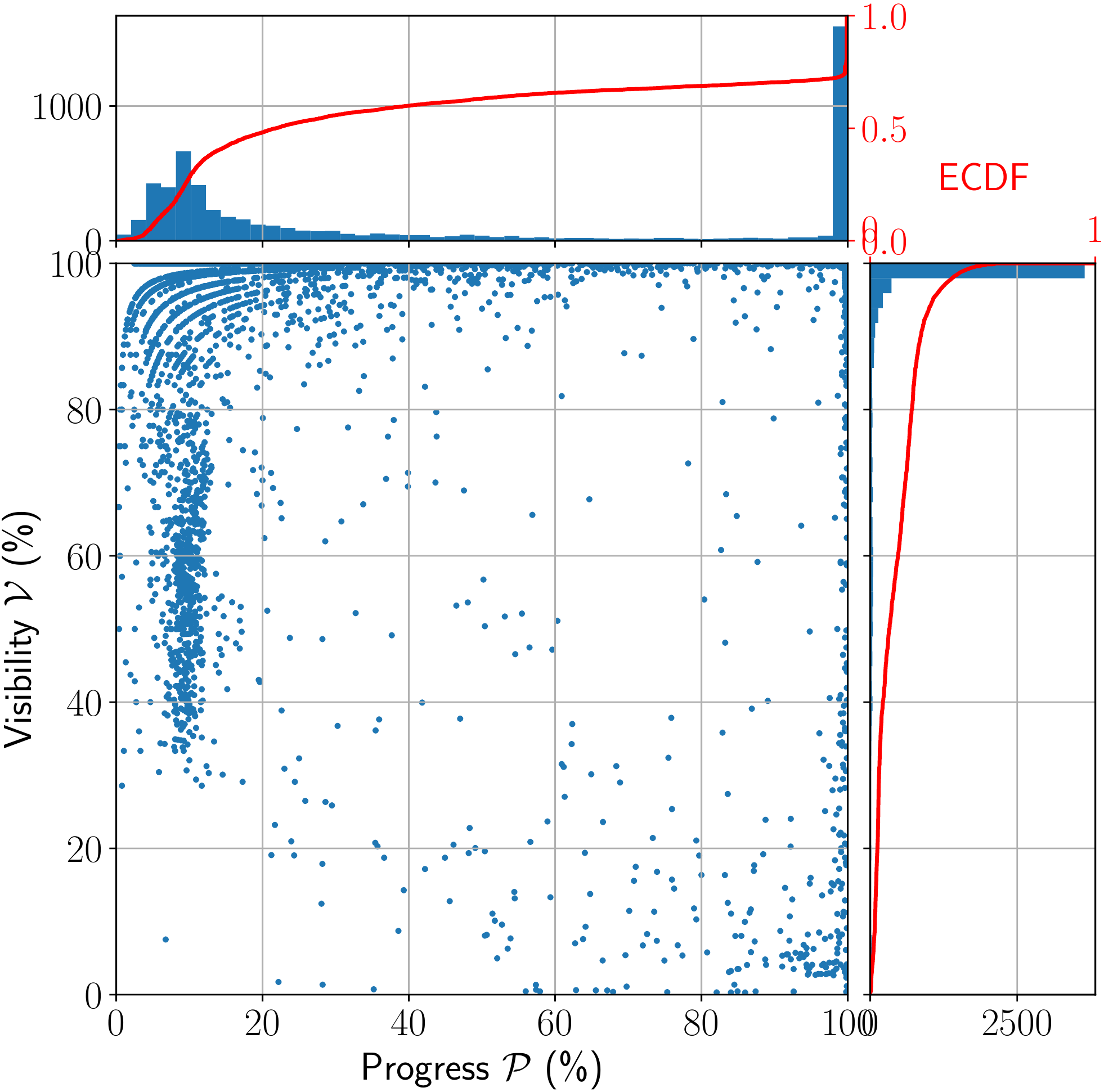}
    \label{fig:visibility_progress_telescope}
  }
  \subfloat[Backbone]{
    \centering
    \includegraphics[width=0.5\textwidth]{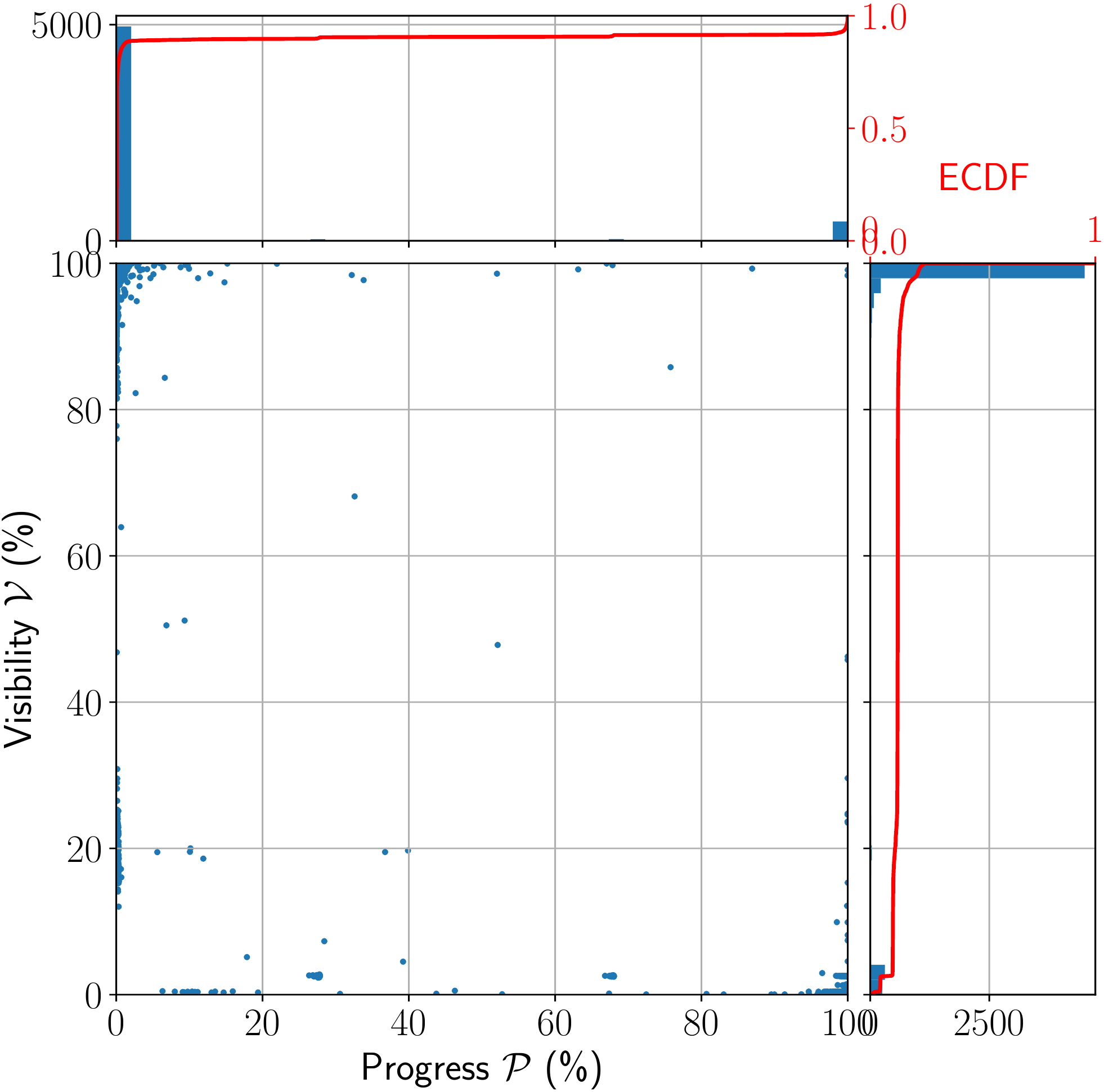}
    \label{fig:visibility_progress_backbone}
  }
  \caption{Scatter plot, histograms and ECDF of visibility $\mathcal{V}$ (percentage of observed probing packets) vs progress $\mathcal{P}$ of scans with the ZMap fingerprint.}
  \label{fig:visibility_progress}
\end{figure*}

\subsubsection{Scan packet visibility and scan progress}
\label{subsec:visibility_progress}

Let $o_1,...,o_m$ the observed probed addresses, and $i_{1},...,i_{m}$ the 
indices of the associated ZMap states, i.e. $o_j = hton(s_{i_j} + offset)$ for 
all $j\in \{1,...,m\}$.
$Det2$ using offset computation is able to identify the generator 
and the prefix 
of the scan.
Then, it is possible to retrieve the internal state and the state index of any 
address $ip$ in $\mathcal{O}$ relatively to the first observed address $o_1$.
Hence, we can reorder the visible IP addresses of the observation network 
$\mathcal{O}$ accordingly to their state indexes, or in other words in their 
scan order.
If $ip \in \mathcal{O}$, we define $rk(ip)$ the rank of the address for this 
order starting from $o_1$ with $rk(o_1)=1$ (see \Cref{fig:notations}). 
Formally, for each $ip \in \mathcal{O}$ let $\delta(ip) = i_{ip} - i_1$ where 
$i_{ip}$ is the index of the internal state that corresponds to $ip$ (i.e. 
$ip = f(s_{i_{ip}})$).
Then, $rk(ip) = |\{ip' \in \mathcal{O} \mid \delta(ip') < \delta(ip) \} |$.

We then compute two metrics.
First, we note $\mathcal{P}$, the progress of the scan.
It is the relative position of the index of the last 
observed probed address within the observed IP addresses: 
$\mathcal{P} = \frac{rk(o_m)}{l}$ where $l$ is the size of $\mathcal{O}$.
This value reflects the scan completion percentage.
A $\mathcal{P}$ value lower than one, means the scan did not reach all IP 
addresses in the observed prefix $\mathcal{O}$.
This may be caused by sharding \cite{Adrian2014Zippier}, ZMap 
capping options (-n for probe number, -N for result number, -t for time
\cite{zmap_rate_limiting_and_sampling}), early interruption by user or 
incomplete data.
Second, we note $\mathcal{V}$ the visibility of the scan.
It is the percentage of observed probing packet: $\mathcal{V} = \frac{m}{rk(o_m)}$.
When this value is lower than one, it means that some packets are missing, 
either due to packet loss or sharding \cite{Adrian2014Zippier}.
On \Cref{fig:notations}, $\mathcal{P} = \frac{5}{10} = 50\%$ and 
$\mathcal{V} = \frac{4}{5} = 80\%$.

\Cref{fig:visibility_progress} present our results.
For the telescope (resp. backbone) dataset, 72\% (resp 85\%) of identified 
scans have a $\mathcal{V}$ greater than 95\%.
Scans with visibility smaller than 95\% exhibit either progress smaller than 
20\% (67\% in the telescope dataset and 30\% in backbone data), or progress 
higher than 90\% (16\% in the telescope dataset and 57\% in backbone data).
We hypothesize that scans with low visibility exhibit low progress 
because they were stopped as soon as users noticed packet loss or upstream 
overload.
Scans with visibility higher than 95\% exhibit a behavior similar to those 
with low visibility regarding progress $\mathcal{P}$.
These scans' progress values are either smaller than 20\% (40.8\% in the 
telescope dataset and 99.7\% in backbone data), or higher than 90\% (34.8\% in 
the telescope dataset and 0.04\% in backbone data).
We notice some peaks for progress value of 5 and 10\% in 
\Cref{fig:visibility_progress_telescope}.
We hypothesize that these values are common for users that do not want to scan 
the full IPv4 address space and use the -n option in ZMap.
Sharding causes visibility values lower than 50\%.
It is visible on the lower right part of \Cref{fig:visibility_progress_telescope}
and on the lower part of \Cref{fig:visibility_progress_backbone}.
We observe many scans with visibility values around 0.4 and 2.5\% in
\Cref{fig:visibility_progress_backbone}.
We hypothesize that these scans uses sharding among 250 and 40 distinct sources.

\begin{table}[t!]
  \setlength\tabcolsep{7pt} %
  \centering
  \captionof{table}{Breakdown of the numbers of IP address that we that identified that used specific generators: 3 or 12.}
  \label{table:generator_specific}
  \begin{tabular}{llrrrrrrrr }
    \toprule
    
    \textbf{Blacklist}
    & \textbf{$g$ value}
    & \multicolumn{2}{c}{\textbf{Telescope}}
    & \multicolumn{2}{c}{\textbf{Backbone}} \\
    
    \cmidrule(lr){3-4}
    \cmidrule(lr){5-6}
    &                    & TCP & UDP & TCP & UDP \\
    
    \midrule
    
    \multirow{2}{*}{$\mathcal{B}_{ZMap}$}
    
    &   3                &   9 & 4 & 3 & - \\
    &  12                &  11 & 4 & 2 & - \\
    
    \midrule
    
    \multirow{2}{*}{$\emptyset$}
    
    &   3                &  9 & 1 & 4 & - \\
    &  12                &  6 & 2 & - & - \\
    
    \bottomrule
  \end{tabular}

\end{table}

\begin{table}[t!]
  \setlength\tabcolsep{4.5pt}
  \centering
  \captionof{table}{Breakdown of generators used by more than one IP.
    The first line means that we identify one group group of 2 IP located in 
    the same /24 with the same $g$ performing TCP scans in the network 
    telescope and 6 group of 2 to 28 IPs .
  }
  \label{table:generator_random}
  \begin{tabular}{llrrrrrrrrrr}
    \toprule

    \multirow{4}{*}{\rotatebox{90}{\parbox{1.cm}{\textbf{Blacklist}}}}
    & \multirow{2}{*}{\parbox{1.2cm}{\textbf{Src IP location}}}
    & \multicolumn{4}{c}{\textbf{Telescope}}
    & \multicolumn{4}{c}{\textbf{Backbone}} \\
    
    \cmidrule(lr){3-6}
    \cmidrule(lr){7-10}
    
    &                    
    & \multicolumn{2}{c}{\textbf{TCP}} 
    & \multicolumn{2}{c}{\textbf{UDP}}
    & \multicolumn{2}{c}{\textbf{TCP}} 
    & \multicolumn{2}{c}{\textbf{UDP}} \\
    
    \cmidrule(lr){3-4}
    \cmidrule(lr){5-6}
    \cmidrule(lr){7-8}
    \cmidrule(lr){9-10}
    
    &
    & \textbf{\#$g$}
    & \textbf{\#IP}
    & \textbf{\#$g$}
    & \textbf{\#IP}
    & \textbf{\#$g$}
    & \textbf{\#IP}
    & \textbf{\#$g$}
    & \textbf{\#IP}
    \\
    
    \midrule
    
    \multirow{7}{*}{\rotatebox{90}{$\mathcal{B}_{ZMap}$}}
    
    & \multirow{3}{*}{Same /24} 
    &  1 & 2 & - &   & 6 & 2-28 & - &   \\
    &                    &  - &   & - &   & 1 &    3 & - &   \\
    &                    &  - &   & - &   & 1 &    4 & - &   \\
    
    \cmidrule(lr){2-10}
    & Several /24        &  - &   & - &   & 8 &   39 &   &   \\
    \cmidrule(lr){2-10}
    & Distinct /24       &  - &   & - &   & - &      & 1 & 2 \\
    \cmidrule(lr){2-10}
    & Mixed              &  1 & 3 & - &   & - &      & - &   \\

    \midrule
    
    \multirow{3}{*}{\rotatebox{90}{$\emptyset$}}
    
    & Same /24           & 17 & 2 & - &   & 6 & 2-28 & - &  \\
    \cmidrule(lr){2-10}
    
    & \multirow{2}{*}{Distinct /24}       
    &  2 & 2 & 3 & 2 & - &      & - &  \\
    &                    &  1 & 3 & - &   & - &      & - &  \\
    \bottomrule
  \end{tabular}

\end{table}

\subsubsection{Generator reuse}

We analyze the generator $g$ (see \Cref{sec:internal_state}) values identified 
by our method.
Generator can be set using the seed option (-e) for repeatability or 
reproducibility purposes, or, to perform sharding (see 
\Cref{sec:internal_state}) from several machines or network interfaces.
We first notice that specific generators (3 and 12) are dominant across 
scans.
\Cref{table:generator_specific} details these occurrences.
We could not find an explanation to the unusual occurrence of these peculiar 
values.
\Cref{table:generator_random} presents generator reuse across IP addresses inside both 
datasets for $g$ other than 3 and 12.
We observe several instance of generator reuse across a number of IP ranging 
from 2 to 39.
This is consistent with the sharding observed on \Cref{fig:visibility_progress}.
Furthermore, we notice that generator reuse mainly occur across IP located 
in the same /24.
We hypothesize that scanning entities usually use several colocated machines 
or VMs from a single provider.
We however notice that only 1.4\% (resp 34\%) of group of source IP in the 
telescope (resp. backbone) dataset sharing generator values have a visibility 
$\mathcal{V}$ value smaller than 50\%.
We thus hypothesize that the majority of generator reuse is linked to setting 
of a constant seed (-e option) for reproducibility purpose, and not sharding.

\begin{figure*}[t!]
  \centering
  Telescope \hspace{6cm} Backbone \\
  \subfloat[TCP]{
    \centering
    \includegraphics[width=0.5\textwidth]{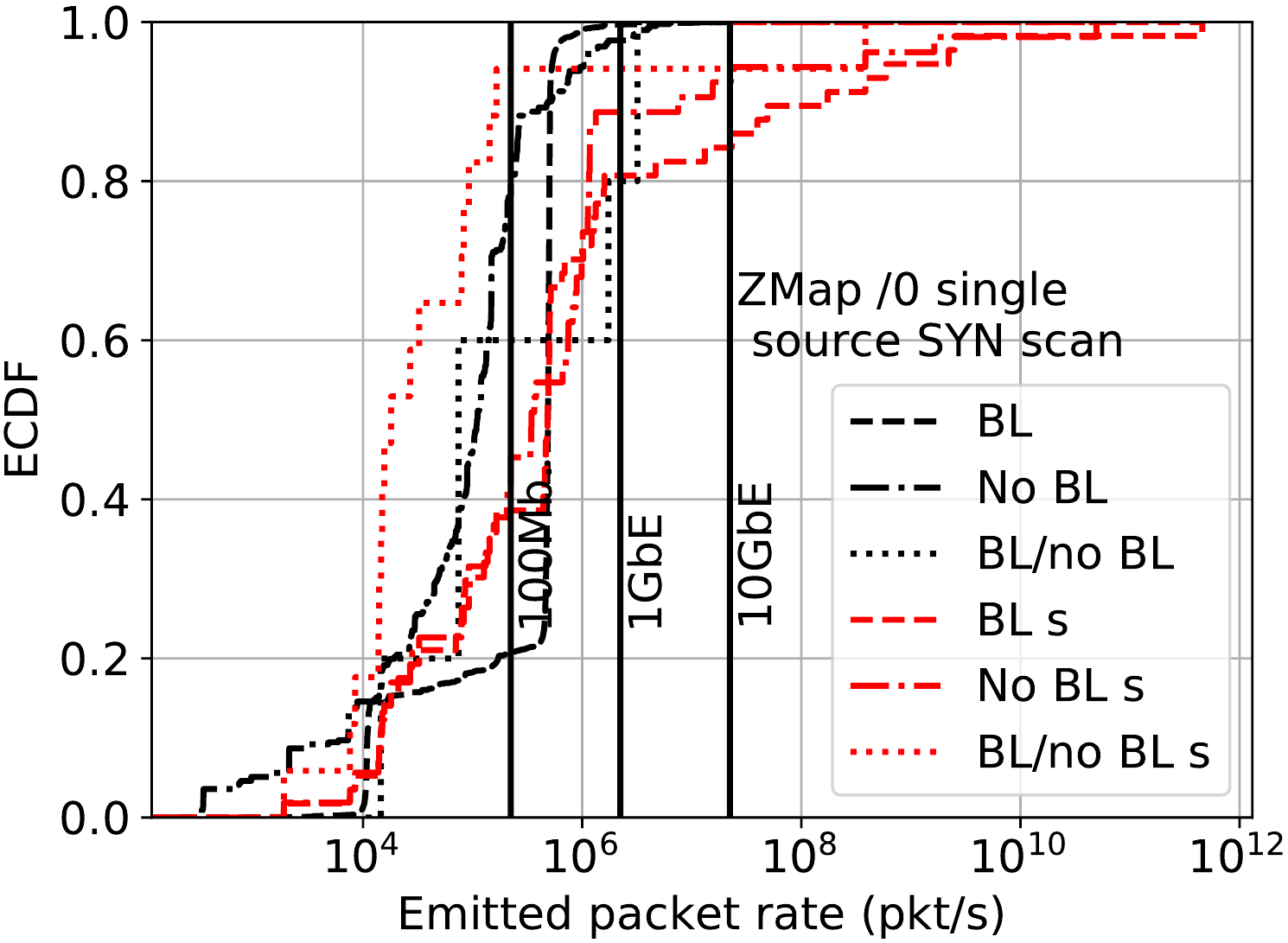}
    \label{fig:packet_rate_ecdf_tcp_sample_telescope}
  }
  \subfloat[TCP]{
    \centering
    \includegraphics[width=0.5\textwidth]{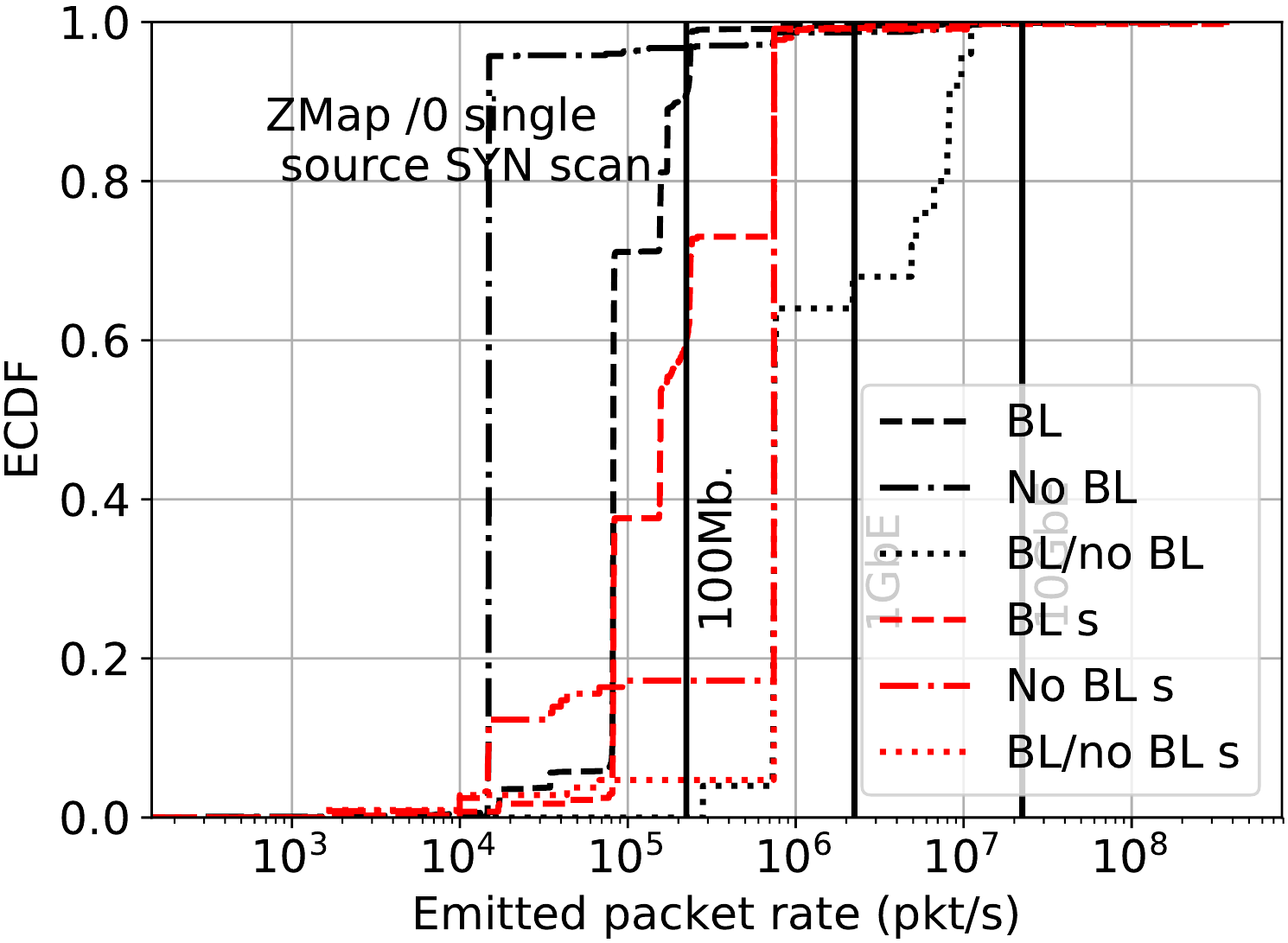}
    \label{fig:packet_rate_ecdf_tcp_sample_backbone}
  }
  \\
  \subfloat[UDP]{
    \centering
    \includegraphics[width=0.5\textwidth]{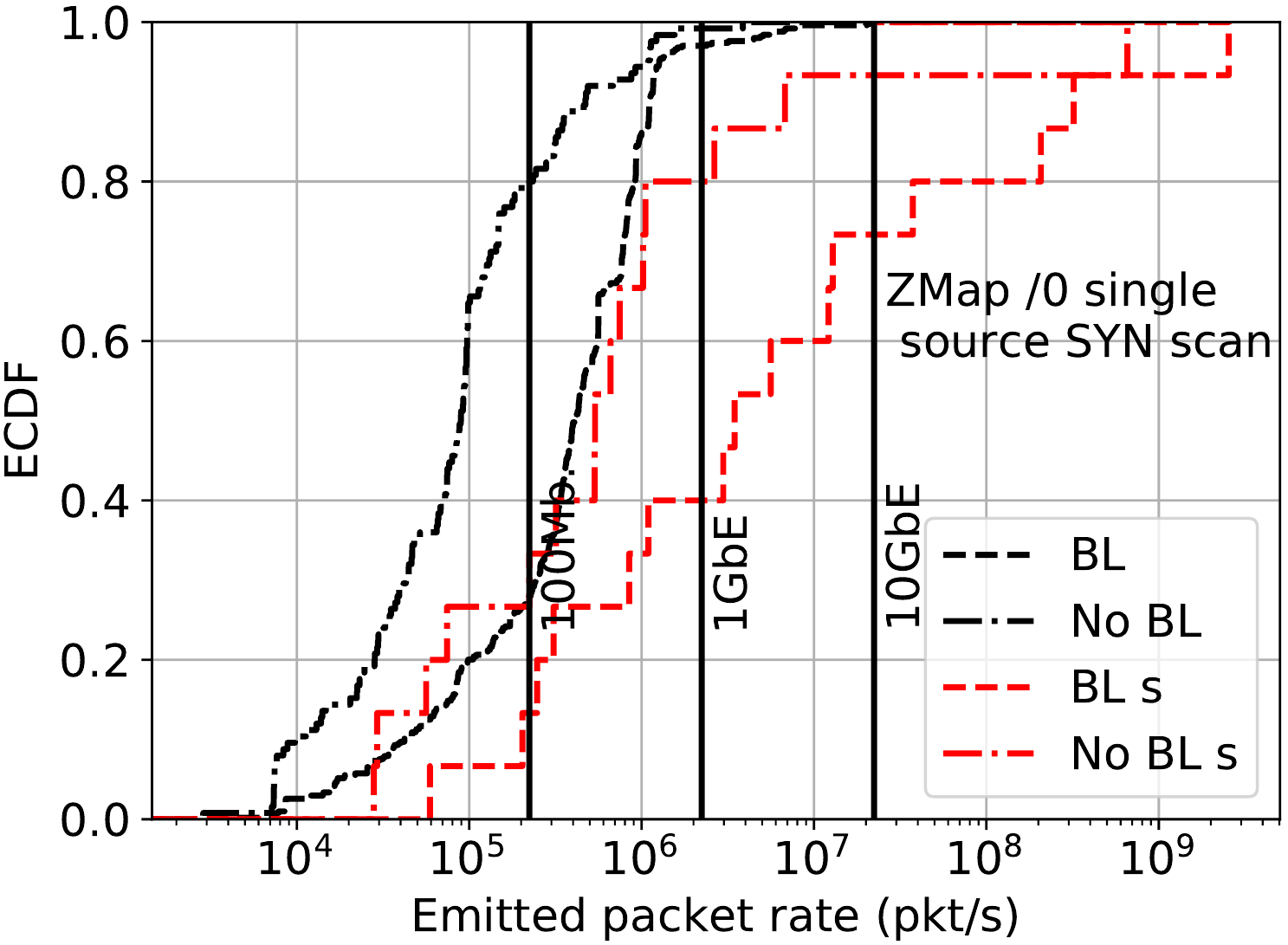}
    \label{fig:packet_rate_ecdf_udp_sample_telescope}
  }
  \subfloat[UDP]{
    \centering
    \includegraphics[width=0.5\textwidth]{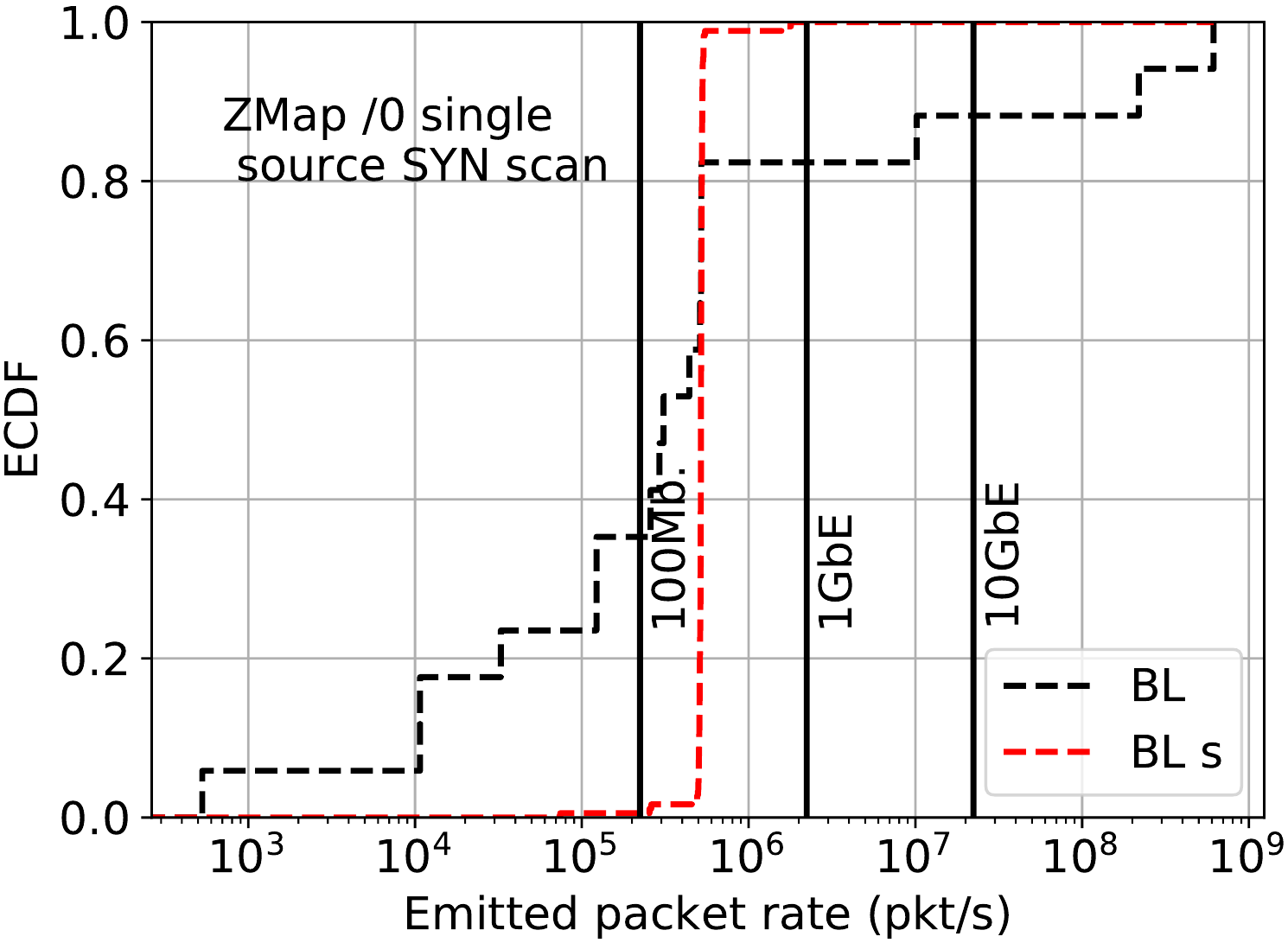}
    \label{fig:packet_rate_ecdf_udp_sample_backbone}
  }
  
  \caption{Emitted packet rate of scans with the ZMap fingerprint and breakdown regarding the use of sampling. 
    Black (resp. red) curves represents scans only identified with or without (resp. only with) the sampled sequence.
    BL: $\mathcal{B}_{ZMap}$ and no BL: $\emptyset$.
  }
  \label{fig:emitted_packet_rate_ecdf_sampling}
\end{figure*}

\subsubsection{Emitted packet rate}
\label{sec:emitted_packet_rate}

The number of packets emitted between the first probe $o_1$ and the last one 
$o_m$ can be retrieved by computing $\delta rk = rk(o_m) - rk(o_1)$.
The elapsed time between the first observed packet and the last is 
$\delta t = t_m - t_1$ (see \Cref{fig:notations}).
We define the emitted packet rate $epr$ as: 
$epr = \frac{\delta rk}{\delta t} = \frac{rk(o_m) - rk(o_1)}{t_m - t_1}$.

\Cref{fig:emitted_packet_rate_ecdf_sampling} depicts emitted packet rate of 
scans.
In order to provide reference, we build $epr$ values for a SYN scan that uses 
100Mb, 1GbE and 10GbE.
We chose the SYN scan because it is the most common type of ZMap scan observed 
\cite{Durumeric2014Security}.
Red curves show scans that were only identified using sampled sequence while 
black ones represents scans identified with or without sampled sequence.
Scans only identified using sampled sequence exhibit a higher emitted packet 
rate than other scans.
We observe that scans identified without any blacklist exhibit a lower packet 
rate than scans detected with the default ZMap blacklist.
Emitted packet rates are usually lower than 100Mb for TCP scans in the backbone 
dataset.
Other scans (TCP and UDP for the telescope dataset and UDP for the backbone 
dataset) are mostly smaller than 1GbE.
1.1\% of scans in the telescope dataset and 0.4\% in the backbone dataset
exhibit $epr$ consistent with a throughput greater than 1GbE.
Their visibility $\mathcal{V}$ (see \Cref{subsec:visibility_progress}) values 
are very small: for the telescope dataset, the median value is 3\% 
and the mean is 4.5\%, for the backbone dataset, these values are 0.04\% and 1.9\%.
We hypothesize that high speed packet rate saturates network equipment 
in the upstream provider of the source IP, yielding packet losses.

\section{Discussions}

Overall, Det2 identification performance in a controlled setup are perfect 
(see \Cref{subsubsec:size_observed_prefix_vs_runtime}).
We however do not identify 58.5\% ZMap scans using TCP and 90.5\% using UDP
in the telescope dataset, and 69.2\% ZMap scans using TCP and 93.1\% using UDP 
in the backbone dataset.
By analyzing reverse DNS records, we find that a significant proportion of
scans with ZMap's IPID originate from entities that use custom blacklists (see \Cref{subsec:general}).
Other possble explanations are the use of several threads (see \Cref{subsec:threads}),
or acute packet reordering (see \Cref{subsec:sampling}).
While threading and packet reordering are linked to ZMap usage and network 
conditions, and thus beyond our reach, customized blacklist can be accounted 
for using offset bruteforce \Cref{subsec:offset_brute_forcing}.
We thus intend to accelerate offset bruteforcing using GPUs.

We identify two biases for our characterization \Cref{subsec:real_world_traffic}.
First, backbone data capture duration is 48 hours.
Computed progress $\mathcal{P}$ (see \Cref{subsec:visibility_progress})
values of scans which were not completely captured may thus be lower than 
their real value.
We however could not access bigger backbone dataset to increase the 
reliability of our data.
Second, we identified and characterized 28.5\% of all scans in our datasets.
Our characterization results may thus be biased towards scans that were 
identified by the $Det2$ method.
We want to improve our identification methods to reduce this bias.

We show in \Cref{subsec:observation_network} that, in the context of the 
exclusion requests documented by Durumeric et al. 
\cite{Durumeric2014Security}, if $\mathcal{O}$ is a /16 subnetwork, then 
probability that $\mathcal{O}$ fulfills the hypothesis regarding blacklisted 
addresses is 77\%.
This probability increases to 98\% if $\mathcal{O}$ is a /20 subnetwork.
The size difference between the backbone dataset and the telescope one may thus 
explain why $Det2$ is less successful on the backbone dataset.
The new sharding mechanism \cite{zmap_sharding_pizza} introduced in September 
2017 impacts some of our metrics.
The progress value $\mathcal{P}$ of a completely observed scan will be 
$\frac{1}{d}$, with $d$ the number of shards.
Our results are not affected by this change because our data was collected 
before it occurred.
Future experiments shall however be careful when using feature like $\mathcal{P}$.

Ruth et al. \cite{Ruth2019Hidden} use ICMP answers to ZMap scans to gather 
information on the Internet's control plane.
We envision that our identification method is able to detect packet 
loss or packet reordering on Internet.
By using several ZMap scan observation points, and by combining this input with 
routing information, it may be possible to locate and diagnose performance 
problems or failures.

\section{Conclusions}

In this work, we propose two cryptogranalysis methods to identify ZMap 
scanning.
We then apply the $Det2$ method with pre-computed offsets to synthetic and real-world traffic, and 
evaluate its efficiency and computing cost.
Finally, we provide an in-depth analysis of ZMap usage in the wild regarding
impact of packet reordering, targeted prefix, generator reuse, probing progress 
and visibility, sharding and probing speed.
We thus identify several misuses, such as private IP address probing and 
packet rate above upstream capacity, that waste network resources.

\subsection*{Acknowledgments}
We would like to thank the WIDE project for granting us access to the MAWI
traces, and Kensuke Fukuda for providing the machines that were used to process 
the network traffic.

\bibliographystyle{plain}
\bibliography{references}

\begin{thebibliography}{10}

\bibitem{rapid7_project_sonar}
{An Introduction to Project Sonar}.
\newblock
  \url{https://web.archive.org/web/20190405162045/https://www.rapid7.com/research/project-sonar/}.
\newblock Accessed: 2019-04-05.

\bibitem{mawi}
{MAWI}.
\newblock
  \url{https://web.archive.org/web/20190116132315/http://mawi.wide.ad.jp/mawi/}.
\newblock Accessed: 2019-01-28.

\bibitem{zmap_default_thread_number}
{Set default of sender threads to 1}.
\newblock
  \url{https://web.archive.org/web/20190128162043/https://github.com/zmap/zmap/commit/c04bc6d060f0ce1f198ff7eda6bf1a685d3b8efa}.
\newblock Accessed: 2019-01-28.

\bibitem{Shodan}
Shodan.
\newblock
  \url{https://web.archive.org/web/20190116132401/https://www.shodan.io/}.
\newblock Accessed: 2019-01-16.

\bibitem{zmap_sharding_pizza}
{Simplify sharding implementation}.
\newblock
  \url{https://web.archive.org/web/20190122170448/https://github.com/zmap/zmap/commit/0bb01d879ebef2c28f038dd032745c72ad5789e9}.
\newblock Accessed: 2019-01-22.

\bibitem{umich}
{University of Michigan}.
\newblock
  \url{https://web.archive.org/web/20190405162042/http://researchscan271.eecs.umich.edu/}.
\newblock Accessed: 2019-04-05.

\bibitem{zmap_blacklisting}
{ZMap - Blacklisting}.
\newblock
  \url{https://web.archive.org/web/20190116132221/https://github.com/zmap/zmap/wiki/Blacklisting}.
\newblock Accessed: 2019-01-16.

\bibitem{zmap_rate_limiting_and_sampling}
{ZMap - Rate Limiting and Sampling}.
\newblock
  \url{https://web.archive.org/web/20190116132111/https://github.com/zmap/zmap/wiki/Rate-Limiting-and-Sampling}.
\newblock Accessed: 2019-01-16.

\bibitem{Adrian2014Zippier}
David Adrian, Zakir Durumeric, Gulshan Singh, and J~Alex Halderman.
\newblock Zippier zmap: Internet-wide scanning at 10 gbps.
\newblock In {\em WOOT}, 2014.

\bibitem{Allman2007IMC}
Mark Allman, Vern Paxson, and Jeff Terrell.
\newblock A brief history of scanning.
\newblock In {\em IMC}, pages 77--82, 2007.

\bibitem{Bennett1999Packet}
Jon~CR Bennett, Craig Partridge, and Nicholas Shectman.
\newblock Packet reordering is not pathological network behavior.
\newblock {\em Transactions on Networking}, 7(6):789--798, 1999.

\bibitem{Brownlee2012PAM}
Nevil Brownlee.
\newblock One-way traffic monitoring with iatmon.
\newblock In {\em PAM}, pages 179--188, 2012.

\bibitem{cohen2013course}
Henri Cohen.
\newblock {\em A course in computational algebraic number theory}, volume 138.
\newblock Springer Science \& Business Media, 2013.

\bibitem{Doerr2016Scan}
Christian Doerr, Mourad el~Maouchi, Sille Kamoen, and Jarno Moree.
\newblock Scan prediction and reconnaissance mitigation through commodity
  graphics cards.
\newblock In {\em CNS}, pages 287--295, 2016.

\bibitem{Durumeric2015Search}
Zakir Durumeric, David Adrian, Ariana Mirian, Michael Bailey, and J.~Alex
  Halderman.
\newblock A search engine backed by {I}nternet-wide scanning.
\newblock In {\em CCS}, pages 542--553, 2015.

\bibitem{Durumeric2014Security}
Zakir Durumeric, Michael Bailey, and J.~Alex Halderman.
\newblock An internet-wide view of internet-wide scanning.
\newblock In {\em USENIX Security}, pages 65--78, 2014.

\bibitem{Durumeric2013ZMap}
Zakir Durumeric, Eric Wustrow, and J~Alex Halderman.
\newblock Zmap: Fast internet-wide scanning and its security applications.
\newblock In {\em USENIX Security}, pages 47--53, 2013.

\bibitem{Falliere2011Distributed}
Nicolas Falliere.
\newblock A distributed cracker for {V}o{IP}.
\newblock
  \url{https://web.archive.org/web/20190131095815/https://www.symantec.com/connect/blogs/distributed-cracker-voip},
  2011.
\newblock Accessed: 2019-01-31.

\bibitem{Gasser2016Scanning}
Oliver Gasser, Quirin Scheitle, Sebastian Gebhard, and Georg Carle.
\newblock Scanning the ipv6 internet: towards a comprehensive hitlist.
\newblock In {\em TMA}, 2016.

\bibitem{gibson2012discrete}
D~Jason Gibson.
\newblock Discrete logarithms and their equidistribution.
\newblock {\em Unif. Distrib. Theory}, 7:147--154, 2012.

\bibitem{Glatz2012IMC}
Eduard Glatz and Xenofontas Dimitropoulos.
\newblock Classifying internet one-way traffic.
\newblock In {\em IMC}, pages 37--50, 2012.

\bibitem{Masscan}
Robert~David Graham.
\newblock {MASSCAN}: Mass ip port scanner.
\newblock
  \url{https://web.archive.org/web/20190116132542/https://github.com/robertdavidgraham/masscan}.
\newblock Accessed: 2019-01-16.

\bibitem{Jaiswal2007Measurement}
Sharad Jaiswal, Gianluca Iannaccone, Christophe Diot, Jim Kurose, and Don
  Towsley.
\newblock Measurement and classification of out-of-sequence packets in a tier-1
  ip backbone.
\newblock {\em Transactions on Networking}, 15(1):54--66, 2007.

\bibitem{Jung2004SP}
Jaeyeon Jung, V.~Paxson, AW. Berger, and H.~Balakrishnan.
\newblock Fast portscan detection using sequential hypothesis testing.
\newblock In {\em Proc. of SP 2004}, pages 211--225, 2004.

\bibitem{Khattak2016You}
Sheharbano Khattak, David Fifield, Sadia Afroz, Mobin Javed, Srikanth
  Sundaresan, Vern Paxson, Steven~J Murdoch, and Damon McCoy.
\newblock Do you see what i see? differential treatment of anonymous users.
\newblock In {\em NDSS}, 2016.

\bibitem{Kleinberg2003Bursty}
Jon Kleinberg.
\newblock Bursty and hierarchical structure in streams.
\newblock {\em Data Mining and Knowledge Discovery}, 7(4):373--397, 2003.

\bibitem{Leonard2013Demystifying}
Derek Leonard and Dmitri Loguinov.
\newblock Demystifying internet-wide service discovery.
\newblock {\em Transactions on Networking}, 21(6):1760--1773, 2013.

\bibitem{Leonard2012Stochastic}
Derek Leonard, Zhongmei Yao, Xiaoming Wang, and Dmitri Loguinov.
\newblock Stochastic analysis of horizontal {IP} scanning.
\newblock In {\em INFOCOM}, pages 2077--2085, 2012.

\bibitem{Mazel2017Profiling}
Johan Mazel, Romain Fontugne, and Kensuke Fukuda.
\newblock Profiling internet scanners: Spatiotemporal structures and
  measurement ethics.
\newblock In {\em TMA}, pages 1--9, 2017.

\bibitem{Mogul1992Observing}
Jeffrey~C Mogul.
\newblock {\em Observing TCP dynamics in real networks}, volume~22.
\newblock 1992.

\bibitem{Murray2012State}
David Murray and Terry Koziniec.
\newblock The state of enterprise network traffic in 2012.
\newblock In {\em APCC}, pages 179--184, 2012.

\bibitem{Paxson1997End}
Vern Paxson.
\newblock End-to-end internet packet dynamics.
\newblock In {\em ACM SIGCOMM Computer Communication Review}, volume~27, pages
  139--152, 1997.

\bibitem{pohlig1978improved}
Stephen Pohlig and Martin Hellman.
\newblock An improved algorithm for computing logarithms over gf (p) and its
  cryptographic significance (corresp.).
\newblock {\em IEEE Transactions on information Theory}, 24(1):106--110, 1978.

\bibitem{Ruth2019Hidden}
Jan R{\"u}th, Torsten Zimmermann, and Oliver Hohlfeld.
\newblock Hidden treasures-recycling large-scale internet measurements to study
  the internet's control plane.
\newblock In {\em PAM}, 2019.

\bibitem{shanks1971class}
Daniel Shanks.
\newblock Class number, a theory of factorization and genera.
\newblock In {\em Proc. Symp. Pure Math, 1971}, volume~20, pages 415--440,
  1971.

\bibitem{Wang2004Study}
Yi~Wang, Guohan Lu, and Xing Li.
\newblock A study of internet packet reordering.
\newblock In {\em ICOIN}, pages 350--359, 2004.

\bibitem{Wustrow2014Tapdance}
Eric Wustrow, Colleen Swanson, and J~Alex Halderman.
\newblock Tapdance: End-to-middle anticensorship without flow blocking.
\newblock In {\em USENIX Security}, pages 159--174, 2014.

\end{thebibliography}

\section*{Supplementary results}
\label{sec:supplementary_material}

\begin{figure*}[h!]
  \centering
  Telescope \hspace{6cm} Backbone \\
  
  \subfloat[ZMap FP - TCP]{
    \centering
    \includegraphics[width=0.5\textwidth]{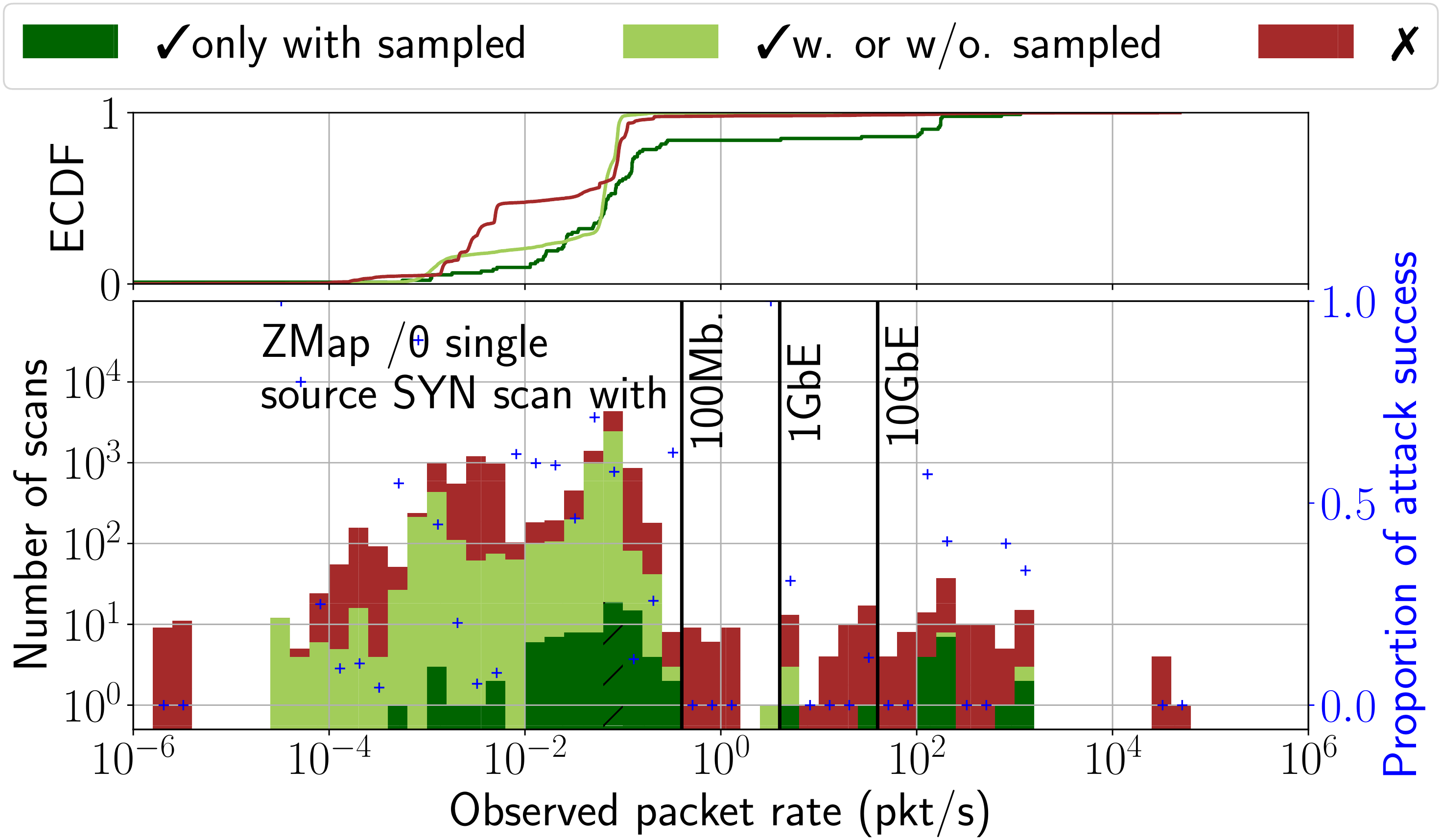}
    \label{fig:observed_packet_rate_zmap_tcp_sampled_dn}
  }
  \subfloat[ZMap FP - UDP]{
    \centering
    \includegraphics[width=0.5\textwidth]{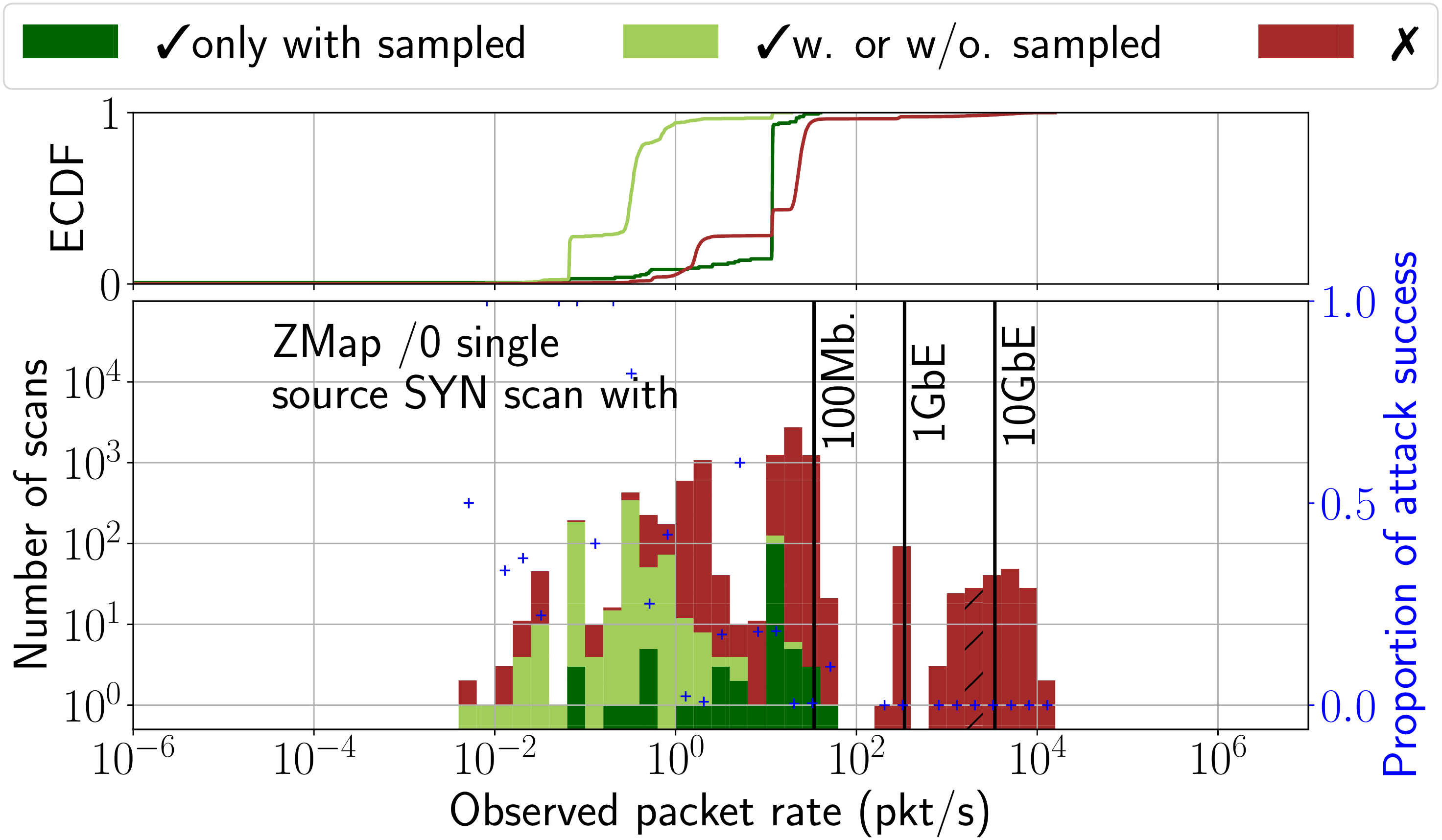}
    \label{fig:observed_packet_rate_zmap_tcp_sampled_mawi}
  }
  \\
  \subfloat[ZMap FP - UDP]{
    \centering
    \includegraphics[width=0.5\textwidth]{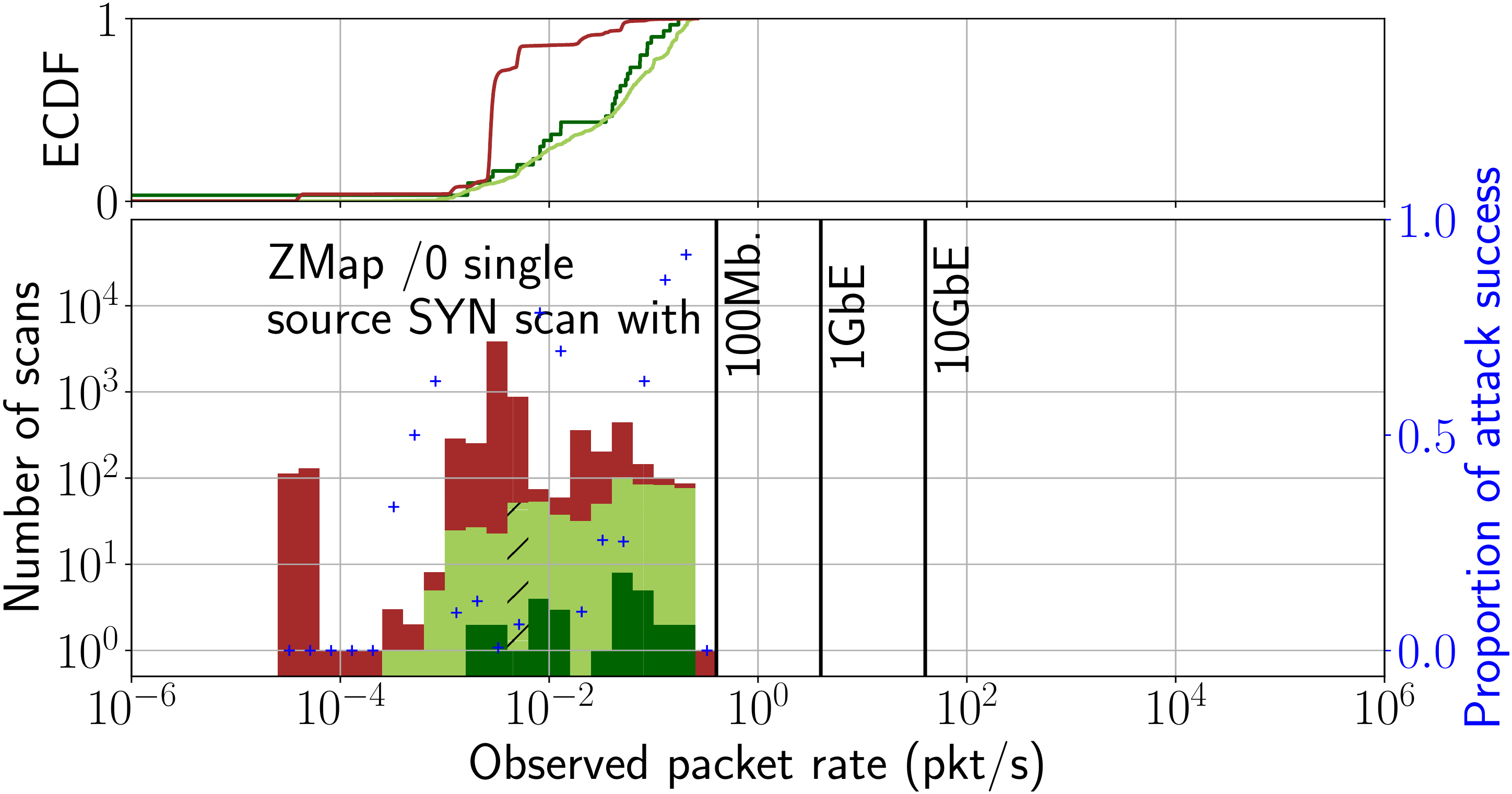}
    \label{fig:observed_packet_rate_zmap_udp_sampled_dn}
  }
  \subfloat[ZMap FP - UDP]{
    \centering
    \includegraphics[width=0.5\textwidth]{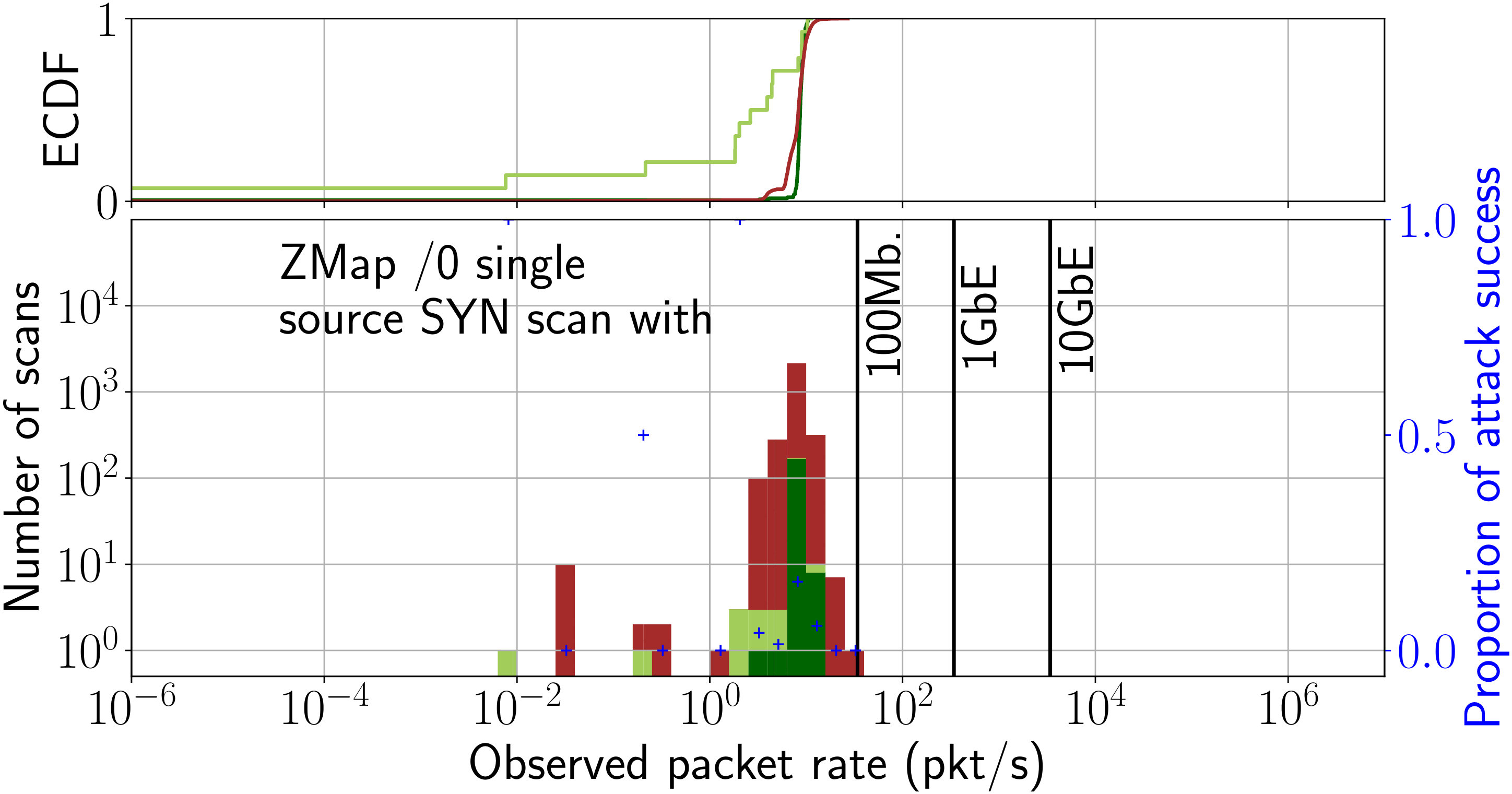}
    \label{fig:observed_packet_rate_zmap_udp_sampled_mawi}
  }
  \\
  \subfloat[Non-ZMap FP - TCP]{
    \centering
    \includegraphics[width=0.5\textwidth]{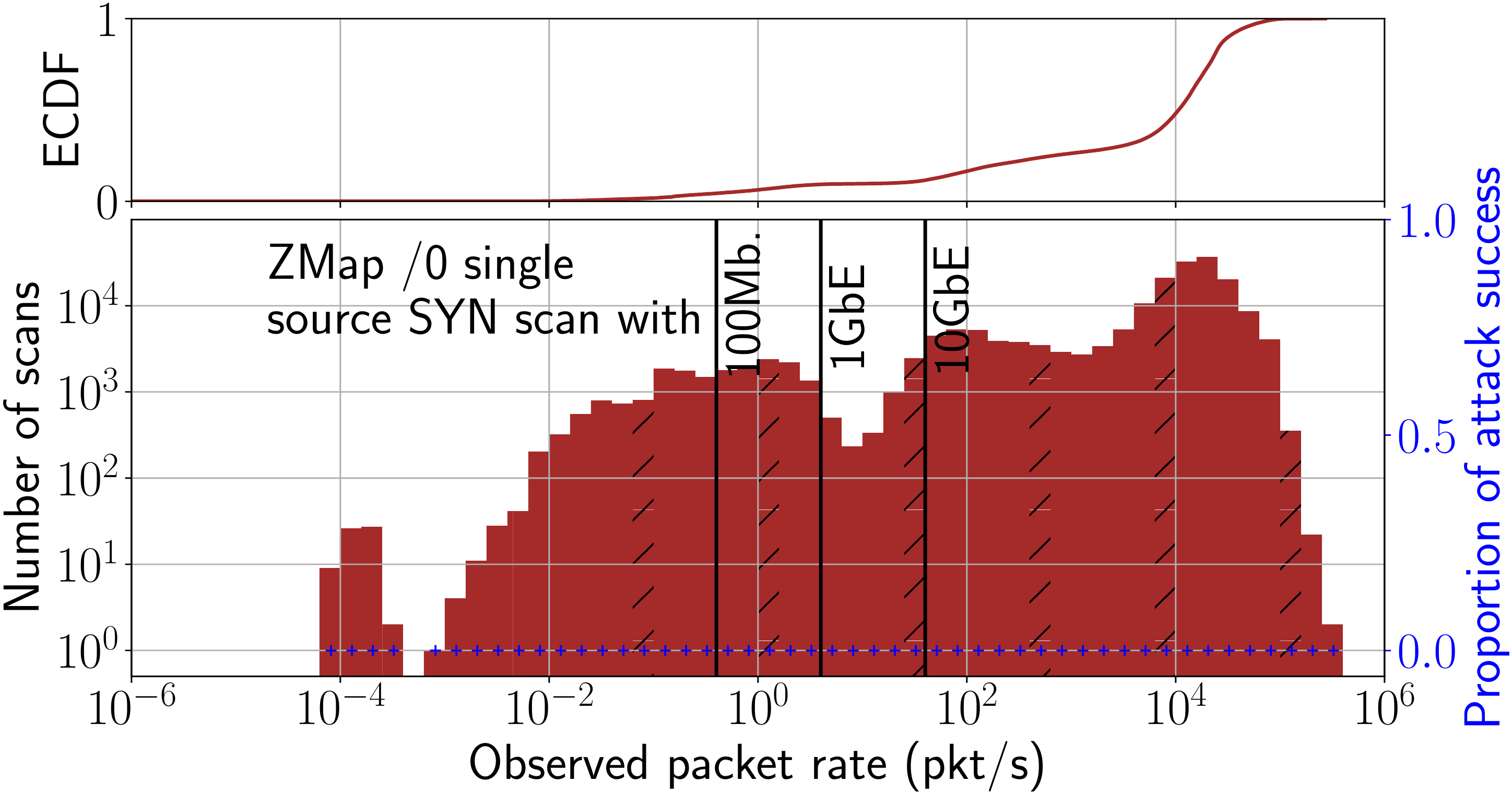}
    \label{fig:observed_packet_rate_non_zmap_tcp_sampled_dn}
  }
  \subfloat[Non-ZMap FP - UDP]{
    \centering
    \includegraphics[width=0.5\textwidth]{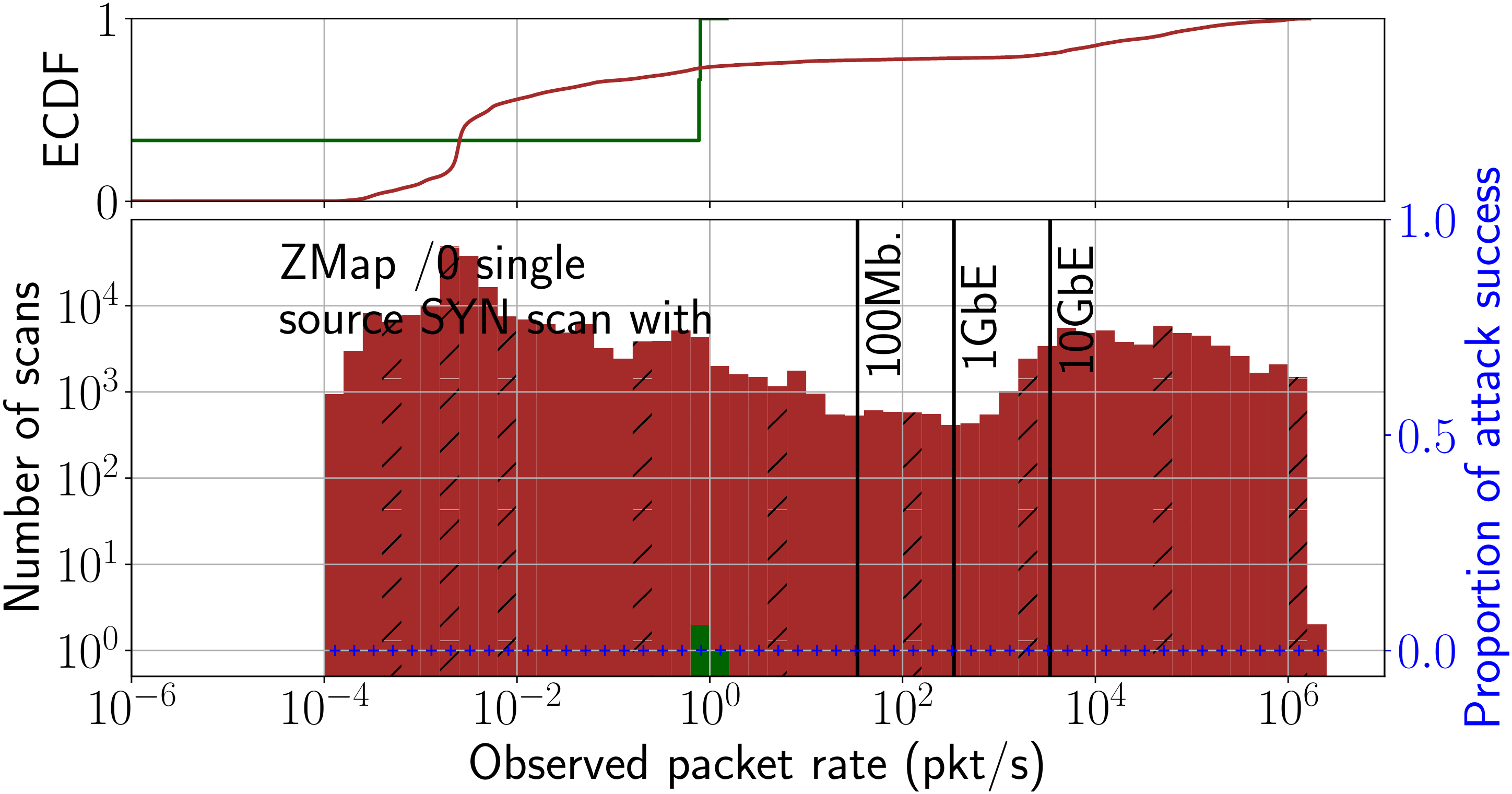}
    \label{fig:observed_packet_rate_non_zmap_udp_sampled_mawi}
  }
  \caption{Stacked histogram of observed packet rate.
    Attack success is green while failure is red.
    Dark green represent scans that were only identified with sampled sequence.
    Light green represent scans were identified with or without sampled sequence.
    Vertical lines represent the theoretical observed packet rates in our three /24 prefixes for a /0 SYN scan performed by ZMap using a single source with 1GbE and 10GbE upstreams.
  }
  \label{fig:observed_packet_rate_tcp_udp_sampled}
\end{figure*}

\subsubsection{Observed packet rate}

The number of observed packets sent by the scanning IP is $m$.
We note the receiving time of $o_1$ (resp. $o_m$), $t_1$ (resp. $t_m$) (see 
\Cref{fig:notations}).
The elapsed time between the first observed packet and the last is 
$\delta t = t_m - t_1$.
We define the observed packet rate $opr$ as: 
$opr = \frac{m}{\delta t} = \frac{m}{t_m - t_1}$.

\Cref{fig:observed_packet_rate_tcp_udp_sampled} presents the observed packet 
rate in our data.
In order to provide reference for $opr$ values, we build $opr$ values for a 
SYN scan that uses standard network speeds such as 100Mb, 1GbE and 10GbE for 
each datasets.
For example, such a scan using 1GbE would generate 0.40 packet per second 
in the telescope dataset, and 34 packets per second in the backbone dataset.
We chose the SYN scan because it is the most common type of ZMap scan observed 
\cite{Durumeric2014Security}.
In both datasets, the observed packet rate of most scans with the ZMap fingerprint 
are consistent with a throughput smaller than 1GbE.
Some scans nonetheless exhibit a throughput bigger than 1GbE.
They may use the same IP on several devices to perform sharding.
We however assess that this behavior is unlikely because it requires more 
configuration compared to standard sharding, and, above all, increases 
detection odds.
We thus hypothesize that these scans actually targets small prefixes.
We also observe that the $Det2$ method is less successful when $opr$ increases.
High throughput may increase packet reordering, and thus degrades 
our identification abilities.
\Cref{fig:observed_packet_rate_tcp_udp_sampled} shows that scans without the 
ZMap fingerprint exhibit a much higher observed rate than scans with the 
fingerprint.
We hypothesize that these scans target smaller network prefixes in 
higher proportion than ZMap scans.

\end{document}